\documentclass[11pt]{article}

\usepackage{graphicx}

\usepackage[colorlinks=true]{hyperref}
\hypersetup{urlcolor=blue, citecolor=red, linkcolor=blue}
\usepackage{amsmath,amssymb}
\usepackage{bm}
\usepackage{enumerate}
\usepackage{amsthm}
\usepackage[all,2cell]{xy}
\usepackage{mathrsfs}
\usepackage{mathtools}
\mathtoolsset{showonlyrefs}
\usepackage{tikz-cd}
\usepackage{xcolor}
\providecommand{\keywords}[1]
{
  \small	
  \textbf{\text{Keywords: }} #1
}
\providecommand{\MSC}[1]
{
  \small	
  \textbf{\text{MSC 2020 codes: }} #1
}

\usetikzlibrary{cd}

\usepackage[a4paper, left=23mm, right=23mm, top=30mm, bottom=30mm]{geometry}

\usepackage{marginnote}

\usepackage{imakeidx}
\makeindex[intoc]

\usepackage[nottoc]{tocbibind}

\theoremstyle{plain}
\newtheorem{theorem}{Theorem}[section]
\newtheorem{corollary}[theorem]{Corollary}
\newtheorem{proposition}[theorem]{Proposition}
\newtheorem{lemma}[theorem]{Lemma}
\newtheorem{example}[theorem]{Example}

\theoremstyle{definition}
\newtheorem{definition}[theorem]{Definition}
\newtheorem{remark}[theorem]{Remark}

\newcommand\restr[2]{{
  \left.\kern-\nulldelimiterspace 
  #1 
  \right|_{#2} 
}}

 \newcommand{\N}{\mathbb{N}}

\newcommand{\R}{\mathbb{R}}
\renewcommand{\d}{\mathrm{d}}

\newcommand{\Cinfty}{\mathscr{C}^\infty}
\newcommand{\T}{\mathrm{T}}
\newcommand{\cT}{\mathrm{T}^\ast}
\newcommand{\Id}{\mathrm{Id}}

\newcommand{\Lie}{\mathscr{L}}

\newcommand{\X}{\mathfrak{X}}
\newcommand{\SL}{\mathrm{SL}}

\newcommand{\A}{\mathcal{A}}
\newcommand{\B}{\mathcal{B}}

\newcommand{\Sym}{\mathrm{Sym}}

\newcommand{\slfr}{\mathfrak{sl}}

\newcommand{\Xh}{\mathfrak{X}_{\rm ham}}

\newcommand{\parder}[2]{\frac{\partial #1}{\partial #2}}
\newcommand{\dparder}[2]{\dfrac{\partial #1}{\partial #2}}
\newcommand{\tparder}[2]{\partial #1/\partial #2}

\DeclareMathOperator{\Ima}{Im}

\DeclareMathAlphabet{\mathpzc}{OT1}{pzc}{m}{it}
\def\d{\mathrm{d}}

\title{{\sffamily Contact Lie systems: theory and applications}}

\author{{\sffamily 
$^a$Javier de Lucas%
\thanks{e-mail:
   javier.de.lucas@fuw.edu.pl \ ORCID: 0000-0001-8643-144X}\ ,\
$^b$Xavier Rivas%
\thanks{e-mail:
   xavier.rivas@unir.net \ ORCID: 0000-0002-4175-5157}
}
\\[2ex]
\normalsize\itshape\sffamily
$^a$UW Institute for Advanced Studies,\\
\normalsize\itshape\sffamily
Department of Mathematical Methods in Physics, University of Warsaw, Warszawa, Poland.
\\[1ex]
\normalsize\itshape\sffamily
$^b$Escuela Superior de Ingenier\'{\i}a y Tecnolog\'{\i}a,\\
\normalsize\itshape\sffamily
Universidad Internacional de La Rioja, Logro\~no, Spain.
\\[1ex]
}

\begin{document}

\maketitle

\begin{abstract}
We define and analyse the properties of contact Lie systems, namely systems of first-order differential equations describing the integral curves of a $t$-dependent vector field taking values in a finite-dimensional Lie algebra of Hamiltonian vector fields relative to a contact manifold. All contact automorphic Lie systems associated with left-invariant contact forms on three-dimensional Lie groups are classified. In particular, we study the so-called conservative contact Lie systems, which are invariant relative to the flow of the Reeb vector field. Liouville theorems, contact Marsden--Weinstein reductions, and Gromov non-squeezing theorems are developed and applied to contact Lie systems. Our results are illustrated by examples with relevant physical and mathematical applications, e.g. Schwarz equations, Brockett systems, quantum mechanical systems, etc. Finally, a Poisson coalgebra method for the determination of superposition rules for contact Lie systems is developed.
\end{abstract}


\MSC{
{\sl Primary:}
37J55; 
53Z05, 
{\sl Secondary:}
34A26, 
34A05, 
17B66, 
22E70. 
}
{\setcounter{tocdepth}{2}
\def\baselinestretch{1}
\small
\def\addvspace#1{\vskip 1pt}
\parskip 0pt plus 0.1mm
\tableofcontents
}

\newpage

\section{Introduction}
A {\it Lie system} is a $t$-dependent system of first-order ordinary differential equations whose general solution can
be expressed via an autonomous function, a {\it superposition rule}, depending on a generic finite family of particular solutions
and some constants to be related to initial conditions \cite{Car2000,CL2011,Win1983}. Examples of Lie systems are Riccati equations and most of their generalisations \cite{CL2011,GL2013,Win1983}.

The Lie–Scheffers theorem says that a Lie system is equivalent  to a $t$-dependent vector field taking
values in a finite-dimensional Lie algebra of vector fields, a so-called {\it Vessiot--Guldberg Lie algebra} (VG Lie algebra, hereafter). This fact illustrates that being a Lie system is the
exception rather than the rule, although Lie systems have numerous and relevant physical
and mathematical applications (see  \cite{CL2011,LS2020} and references therein).

Lie systems admitting a VG Lie algebra of Hamiltonian vector fields relative to different geometric structures have been studied in recent years (see \cite{LS2020} for a survey on the topic). In particular, \cite{BBHLS13,BCHLS13,Car2000,CLS13,Car2019,Ru10} analyse Lie systems possessing a VG Lie algebra of Hamiltonian vector fields relative to a Poisson structure: the so-called {\it Lie--Hamilton systems} (see \cite{CLS13} for symplectic cases). Meanwhile, \cite{Car2014} provides a no-go theorem showing that  Lie--Hamilton systems cannot be used to describe certain Lie systems and shows that, sometimes, one may consider their VG Lie algebras to consist of Hamiltonian vector fields relative to a Dirac structure, which in turn allows one to use Dirac geometry to study their properties. Additionally,  $k$-{\it symplectic Lie systems}, i.e. Lie systems admitting a Vessiot--Guldberg Lie algebra of Hamiltonian vector fields relative to a $k$-symplectic manifold, were analysed in \cite{LV15}. Meanwhile, {\it multisymplectic Lie systems}, along with a certain type of multisymplectic reduction, were studied in \cite{GLMV19,GLRRV22}. It is quite interesting that finding Lie systems with VG Lie algebras of Hamiltonian vector fields relative to some geometric structure has led to a bloom in the description of new applications of Lie systems, despite being differential equations satisfying more restrictive conditions than mere classical Lie systems \cite{BCHLS13,BHLS15,Car2019,LS2020}. It is remarkable that geometric structures allow for the construction of superposition rules, constants of motion, and the analysis of relevant properties of Lie systems without relying on the analysis/solution of systems of partial or ordinary differential equations as the most classical and old methods \cite{Car2000,CGM07,CL2011,LS2020,Win1983}. Geometric techniques also provide new viewpoints to the nature and properties of superposition rules \cite{BCHLS13} and mathematical/physical problems \cite{Car2019,LL18}.

In this context, this work investigates Lie systems possessing a VG Lie algebra of Hamiltonian vector fields relative to a contact structure, the referred to as {\it contact Lie systems}. Contact Lie systems can be considered as a particular case of {\it Jacobi--Lie systems} (see \cite{AHKR22,AHR22,HLS15}), which were first introduced in \cite{HLS15}. Nevertheless, \cite{HLS15} just contained one non-trivial example of Jacobi--Lie system giving rise to a contact Lie system and it did no analyse the properties that are characteristic for contact Lie systems. In fact, \cite{HLS15} was mostly dealing with Jacobi--Lie systems on one- and two-dimensional manifolds, which only retrieve the trivial contact Lie systems with a zero- or one-dimensional VG Lie algebra on $\mathbb{R}$. 

As a particular case, this work analyses the hereafter called {\it conservative contact Lie systems}, namely contact Lie systems that are invariant relative to the Reeb vector field of their associated contact manifolds. For these systems, we introduce certain Liouville theorems, Marsden--Weinstein reductions, and Gromov non-squeezing theorems, whose application can be considered as pioneering in the literature of Lie systems. Moreover, it is remarkable that the literature on contact systems is mostly focused on dissipative systems \cite{Bra2018,Bra2017b,Cia2018,Gas2019,DeLeo2020,DeLeo2019b}. Meanwhile, this work also treats contact Hamiltonian systems not related to dissipation while having physical applications. 

Willet's reduction of contact manifolds \cite{Wil2002} is applied to the reduction of contact Lie systems. This is more general than some other momentum maps reductions appearing in the literature \cite{DeLeo2019}. It is worth noting that types of Marsden--Weinstein reductions have been applied to Lie systems in \cite{GLRRV22} for multisymplectic Lie systems.

Finally, an adaptation of the coalgebra method to obtain superposition rules, which was firstly aimed at  Lie--Hamilton and Dirac--Lie systems \cite{CLS13,Car2014}, has been devised for Jacobi--Lie systems, and therefore, contact Lie systems. To illustrate our methods, an application to derive a superposition rule for an automorphic Lie system on the Lie group $\SL(2,\mathbb{R})$ has been developed. 

Although contact Lie systems are naturally related to symplectic Lie systems on manifolds of larger dimension, this relation is shown to be, to our purposes, rather a mere curiosity without practical applications. For instance, the latter appears as a byproduct of our classification of contact {\it automorphic Lie systems} (see \cite{LS2020} for a definition) on three-dimensional Lie groups. It is worth noting that the relevance of automorphic Lie systems is due to the fact that the solution of every Lie system can be obtained by means of a particular solution of an automorphic Lie system and the integration of a VG Lie algebra to a Lie group action \cite{CL2011}. Other relations of contact Lie systems with multisymplectic Lie systems and Jacobi--Lie systems are discussed. Although it is shown that contact Lie systems can be considered as particular cases of the above-mentioned types of Lie systems, it is stressed that contact Lie systems have natural properties, e.g. associated volume forms or reductions, that are more properly studied in the context of contact geometry.

The structure of the work goes as follows. In Section \ref{Se:RevConMec}, a review on contact geometry and contact Hamiltonian systems is provided and Willet's reductio on contact manifolds is sketched. Section \ref{Se:ConLieSys} is the theoretical core of the article, introducing the notion of contact Lie system and of conservative contact Lie system. Moreover a Gromov’s non-squeezing theorem for conservative contact Lie systems is stated and proved. In Section \ref{subsec:other}, we analyse existence of underlying geometric structures for contact Lie systems. Section \ref{Se:ExiCon} classifies a class of contact Lie systems admitting a Vessiot--Guldberg Lie algebra of right-invariant vector fields \cite{LS2020} on three-dimensional Lie groups.  Section \ref{Se:Examples} is devoted to presenting four examples: the Brockett control system, the Schwarz equation, a family of quantum contact Lie systems, and a contact Lie system that is not conservative. Finally, Section \ref{Se:Coalgebra} describes the coalgebra method for obtaining superposition rules for Jacobi--Lie systems, which gives, in particular, techniques to obtain superposition rules for contact Lie systems. As an application, the superposition rule for an automorphic Lie system on $\SL(2,\mathbb{R})$ is retrieved.

\section{Review on contact mechanics}\label{Se:RevConMec}

From now on, all manifolds and mappings are assumed to be smooth and connected, unless otherwise stated. This will be used to simplify our presentation while stressing its main points. The space of vector fields on a manifold $M$ is denoted by $\X(M)$ while $\Omega^1(M)$ stands for the space of differential one-forms on $M$. Einstein notation will be hereafter used.

\subsection{Contact Hamiltonian systems}

Let us provide a brief introduction to contact geometry (see \cite{Ban2016,Gei2008,Kho2013} for details). In particular, we will also show why, although contact manifolds can be described via some other structures, such approaches are not appropriate for our purposes in this work.

Let us recall that a distribution $\mathcal{D}$ of corank one on a smooth manifold $M$ is called {\it maximally non-integrable} if around every $x\in M$ there exist two locally defined vector fields $X,Y$ taking values in $\mathcal{D}$ such that $[X,Y]_x\notin \mathcal{D}_x$.
A \textit{contact manifold} is a pair $(M,\xi)$ such that $M$ is a $(2n+1)$-dimensional manifold $M$ and $\xi$ is a one-codimensional maximally non-integrable distribution on $M$. We call $\xi$ a \textit{contact distribution} on $M$. Note that $\xi$ can locally be, on an open neighbourhood $U$ of each point $x\in M$,  described as the kernel of a one-form $\eta\in\Omega^1(U)$ such that $\eta\wedge(\d\eta)^n$ is a volume form on $U$ for some $n\in \mathbb{N}\cup\{0\}$. 

A {\it co-orientable contact manifold} is a pair $(M,\eta)$, where $\eta$ is a differential one-form on $M$ such that $(M,\ker \eta)$ is a contact manifold. Then, $\eta$ is called a {\it contact form}. Since this work focus on local properties of contact manifolds and related structures, we will hereafter restrict ourselves to co-oriented contact manifolds. To simplify the notation, co-oriented contact manifolds will be called contact manifolds as in the standard literature on contact geometry \cite{Bra2017a,Gas2019,DeLeo2019b}.

Note that if $\eta$ is a contact form on $M$, then $f\eta$ is also a contact form on $M$ for every non-vanishing function $f\in\Cinfty(M)$. Moreover,  $\eta\wedge(\d\eta)^n$ is a volume form on $M$ if and only if $\eta$ induces a decomposition of the tangent bundle of the form $\T M = \ker\eta\oplus\ker\d\eta$.

A contact manifold $(M,\eta)$ determines a unique vector field $R\in\X(M)$, called the {\it Reeb vector field}, such that $i(R)\d\eta = 0$ and $i(R)\eta = 1$. Then, $\Lie_R\eta = 0$ and, therefore, $\Lie_R\d\eta = 0$.

\begin{theorem}{\bf (Darboux theorem)}
    Given a contact manifold $(M,\eta)$, where $\dim M=2n+1$, every point $x\in M$ admits a local open coordinated neighbourhood  with coordinates $\{q^i, p_i, s\}$, with $i = 1,\dotsc,n$, called {\it Darboux coordinates},  such that
    $$ \eta = \d s - p_i\d q^i\,. $$
    In these coordinates, $R = \tparder{}{s}$.
\end{theorem}

A proof of the Darboux theorem for contact manifolds can be found in \cite{Abr1978,Lib1987}.

\begin{example}{\bf(Canonical contact manifold)}
Consider the product manifold $M = \cT Q\times\R$, where $Q$ is any manifold. The cotangent bundle $\cT Q$ admits an adapted coordinate system $\{q^1,\ldots,q^n,p_1,\ldots,p_n\}$ and $\R$ has a natural coordinate $s$, which in turn give rise to a natural coordinate system $\{q^1,\ldots,q^n, p_1,\ldots,p_n, s\}$ on $\cT Q\times \R$. The one-form $\eta = \d s - \theta$, where $\theta$ is the pull-back of the Liouville one-form $\theta_\circ\in\Omega^1(\cT Q)$ relative to the canonical projection $\cT Q\times \R\to\cT Q$. In the chosen  coordinates, 
$$ \eta = \d s - p_i\d q^i\,,\qquad R=\parder{}{s}\,.$$
    The coordinates $\{q^i,p_i,s\}$ are Darboux coordinates on $M$. It is remarkable that $\theta_\circ$, and thus $\eta$, are independent of the coordinates $\{q^1,\ldots,q^n\}$. 
\end{example}

The previous example is a particular case of {\it contactification} of an exact symplectic manifold. Given an {\it exact symplectic manifold} $(N,\omega)$, namely a symplectic manifold whose symplectic form, $\omega$, is exact, e.g. $\omega = -\d\theta$, the product manifold $M = N\times\R$ is a contact manifold with the contact form $\eta = \d s - \theta$, where the variable $s$ stands for the canonical coordinate in $\R$.

Let $(M,\eta)$ be a contact manifold. There exists a vector bundle isomorphism $\flat:\T M\to\cT M$ given by
$$ \flat(v) = i(v)(\d\eta)_x + (i(v)\eta_x)\eta_x\,,\qquad \forall v\in \T_x M,\quad\forall x\in M\,. 
$$
This isomorphism can be extended to a $\Cinfty(M)$-module isomorphism $\flat:\X(M)\to\Omega^1(M)$ in the natural manner. It is usual to denote both isomorphisms, of vector bundles and of $\Cinfty(M)$-modules, by $\flat$ as this does not lead to misunderstanding. Taking into account this isomorphism, $R = \flat^{-1}(\eta)$. The inverse of $\flat$ is denoted by $\sharp = \flat^{-1}$.


A {\it contact Hamiltonian system} \cite{Bra2017b,Gas2019,DeLeo2019b} is a triple $(M,\eta,h)$, where $(M,\eta)$ is a contact manifold and $h\in\Cinfty(M)$. Given a contact Hamiltonian system $(M,\eta,h)$, there exists a unique vector field $X_h\in\X(M)$, called the {\it contact Hamiltonian vector field} of $h$, satisfying the following equivalent conditions
\begin{enumerate}[(1)]
    \item $\ i(X_h)\d\eta = \d h - (\Lie_R h)\eta\quad$ and $\quad i(X_h)\eta = -h$,
    \item $\ \Lie_{X_h}\eta = -(\Lie_R h)\eta\quad$ and $\quad i(X_h)\eta = -h$,
    \item $\ \flat(X_h) = \d h - (\Lie_R h + h)\eta$.
\end{enumerate}
A vector field $X\in\X(M)$ is said to be \textit{Hamiltonian} relative to the contact structure given by $\eta$ if it is the Hamiltonian vector field of some function $h\in\Cinfty(M)$. Let $\Xh(M)$ stand for the space of Hamiltonian vector fields relative to $(M,\eta)$. Unlike in the case of symplectic mechanics, the Hamiltonian function $h$ is not preserved under the evolution of the contact Hamiltonian vector field $X_h$ (see \cite{Gas2019,Mars1988} for details). More precisely,
$$ \Lie_{X_h}h = -(\Lie_R h)h\,. $$
In Darboux coordinates, the contact Hamiltonian vector field $X_h$ reads
\begin{equation}\label{Eq:HamCoor} X_h = \parder{h}{p_i}\parder{}{q^i} - \left( \parder{h}{q^i} + p_i\parder{h}{s} \right)\parder{}{p_i} + \left( p_i\parder{h}{p_i} - h \right)\parder{}{s}\,. \end{equation}
Its integral curves, let us say $\gamma(t) = (q^i(t), p_i(t), s(t))$, satisfy the system of differential equations
$$
    \frac{\d q^i}{\d t} = \parder{h}{p_i}\,,\qquad
    \frac{\d p_i}{\d t} = - \left( \parder{h}{q^i} + p_i\parder{h}{s} \right)\,,\qquad
    \frac{\d s}{\d t} = p_j\parder{h}{p_j} - h\,,\qquad i = 1,\dotsc,n\,.
$$
\begin{example}
    Consider the contact Hamiltonian system $(\cT \R^n\times\R, \eta, h)$, where $\R^n$ has linear coordinates $\{q^1,\ldots,q^n\}$, while $\eta = \d s - p_i\d q^i$ and
    $$ h = \frac{p^2}{2m} + V(q) + \gamma s\,, $$
    where $m$ is the mass of a particle, $p = \sqrt{p_1^2 + \dotsb + p_n^2}$, $\gamma\in \mathbb{R}$, and $V(q)$ is a potential. The Hamiltonian function $h$ describes a mechanical system consisting of a particle under the influence of a potential $V(q)$ and with a friction force proportional to the momenta. The integral curves of the contact Hamiltonian vector field, $X_h$, satisfy the system of equations
    $$ \frac{\d q^i}{\d t} = \frac{p_i}{m}\,,\qquad \frac{\d p_i}{\d t} = -\parder{V}{q^i}(q) - \gamma p_i\,,\qquad \frac{\d s}{\d t} = \frac{p^2}{2m} - V(q) - \gamma s\,,\qquad i = 1,\dotsc,n\,.$$
    Combining the first two equations, one gets
    $$ m\frac{\d^2 q^i}{\d t^2} + \gamma m \frac{\d q^i}{\d t} + \parder{V}{q^i}(q) = 0\,,\qquad i=1,\dotsc,n\,.
    $$
\end{example}
\bigskip

Finally, let us recall that a contact manifold $(M,\eta)$ gives rise to a Lie bracket
\begin{equation}\label{Eq:ContactBrack}
\{f,g\}=X_fg+gRf=-\d\eta (X_f,X_g)-fRg+gRf\,,\qquad \forall f,g\in \Cinfty(M)\,.
\end{equation}
In view of \eqref{Eq:HamCoor}, one can prove that
$$
\{f,gh\} = h\{f,g\} + g\{f,h\} + ghRf\,,\qquad \forall f,g,h\in \Cinfty(M)\,.
$$
Hence, \eqref{Eq:ContactBrack} is a Poisson bracket if and only if $R = 0$, which is a contradiction. Nevertheless, note that, if $\Cinfty_g(M)$ stands for the space of good Hamiltonian functions, then the restriction of $\{\cdot,\cdot\}$ to $\Cinfty_g(M)$ becomes a Poisson bracket.

The formalism presented in this section has a Lagrangian counterpart \cite{Gas2019,PhDThesisXRG}. In addition, a geometric formulation for time-dependent contact systems developing the so-called cocontact geometry has been introduced in \cite{LGGMR-2022,RiTo-2022}.

\subsection{Contact manifolds and other geometric structures}

Let us study several geometric structures used to describe particular aspects of contact manifolds.

\begin{definition}
    A \textit{Jacobi manifold} is a triple $(M,\Lambda,E)$, where $\Lambda$ is a bivector field on $M$, i.e. a skew-symmetric 2-contravariant tensor field, and $E$ is a vector field on $M$, such that
    $$ [\Lambda,\Lambda] = 2E\wedge\Lambda\ ,\qquad \Lie_E\Lambda = [E,\Lambda] = 0\,, $$
    where $[\cdot,\cdot]$ denotes the Schouten--Nijenhuis bracket in its original sign convention\footnote{There exists a modern, and sometimes more appropriate, definition of the Schouten--Nijenhuis bracket that differs from ours on a global proportional sign depending on the degree of $\Lambda$ (see Example 2.20 in \cite{Grab-2013} and references therein).} \cite{Nij1955,Sch1953,Va94}.
\end{definition}

\begin{remark}
    In particular, \textit{Poisson manifolds} are equivalent to Jacobi manifolds with $E = 0$. In turn, Poisson manifolds retrieve, as particular cases, symplectic and cosymplectic manifolds \cite{CN13,Va94}.
\end{remark}

Every bivector field $\Lambda$ on $M$ induces a vector bundle morphism $\Lambda^\sharp:\cT M\rightarrow \T M$ given by $\Lambda^\sharp(\vartheta_x) = \Lambda_x(\vartheta_x,\cdot)$ for every $\vartheta_x\in \cT_x M$ with $x\in M$.

A {\it Hamiltonian vector field} relative to $(M,\Lambda,E)$ is a vector field $X$ on $M$ of the form 
$$
    X = \Lambda^\sharp (\d f) + fE\,,
$$
for a function $f\in\Cinfty(M)$, which is called a {\it Hamiltonian function} of $X$. It can be proved that if $E_x\notin {\rm Im} \,\Lambda^\sharp_x$ at every point $x\in M$, then each Hamiltonian vector field has a unique Hamiltonian function. Additionally, $X$ is called a {\it good Hamiltonian vector field} if it admits a Hamiltonian function $f$ satisfying $Ef = 0$. 

The {\it characteristic distribution} of $(M,\Lambda,E)$ is the generalised distribution \cite{Va94} on $M$ of the form
$$
    \mathcal{C}_x = \Ima\Lambda_x^\sharp +\langle E_x\rangle,\qquad \forall x\in M\,.
$$
It can be proved that the characteristic distribution of a Jacobi manifold $(M,\Lambda,E)$ is integrable and its maximal integral submanifolds are such that, if even-dimensional, then $\Lambda$ gives rise to a {\it locally conformal symplectic form}, while if the maximal integral submanifold is odd-dimensional, then $\Lambda$ gives rise to a contact manifold \cite{Va94}. Recall that a contact manifold $(M,\eta)$ with Reeb vector field $R$ gives rise to a Jacobi manifold $(M,\Lambda,-R)$, where $\Lambda$ is the bivector field such that $\Lambda^\sharp$ is equal to the isomorphism $\sharp = \flat^{-1}:\cT M\rightarrow \T M$ (see \cite{DeLeo2019b,Gas2019}). Moreover, every Jacobi manifold $(M,\Lambda,E)$ gives rise to a Jacobi bracket given by
$$
    \{f,g\}=\Lambda(\d f,\d g)+fEg-gEf\,,\qquad \forall f,g\in \Cinfty(M)\,.
$$
It is important to remark that the bracket above is not a Poisson bracket. Moreover, $\{\cdot,\cdot\}$ becomes a Poisson bracket when restricted to the space of good Hamiltonian functions, $\Cinfty_g(M)$, of the Jacobi manifold $(M,\Lambda,E)$.

In particular, the Jacobi bracket satisfies
$$ \{f,g\} = X_fg - gEf\,, $$
and it matches the definition of the bracket for contact manifolds when $(M,\eta)$ is such that
$$ \Lambda (\d f,\d g) = -\d\eta(\Lambda^\sharp(\d f),\Lambda^\sharp(\d g))\,,\qquad E = -R\,. $$
The space $\Xh(M)$ of Hamiltonian vector fields in a Jacobi manifold is a Lie algebra with respect to the Lie bracket of vector fields. More precisely, if $X_f,X_g\in\X(M)$ are the Hamiltonian vector fields related to two arbitrary functions $f,g\in\Cinfty(M)$ respectively, one has
$$ [X_f,X_g] = X_{\{f,g\}}\,. $$

\subsection{Reduction of contact manifolds}

Let us describe the contact Marsden--Weinstein reduction theory \cite{Wil2002}. First, let us remark a non-standard fact about momentum mappings for contact manifolds, which make them special and it will have important consequences hereafter.

\begin{definition}
    Let $\Phi:G\times M\rightarrow M$ be a Lie group action preserving the contact form, $\eta$, of a contact manifold $(M,\eta)$, i.e. $\Phi^\ast_g\eta = \eta$ for every $g\in G$. We call $\Phi$ a {\it contact Lie group action}. A \textit{contact momentum map} associated with $\Phi$ is a map $J\colon M\to\mathfrak{g}^\ast$ defined by
    $$ \langle J(x), \xi\rangle = i_{\tilde\xi_x}\eta_x\,,\qquad \forall x\in M\,, $$
    where $\tilde\xi\in\X(M)$ is the fundamental vector field\footnote{We define the fundamental vector field of $\Phi:G\times M\rightarrow M$ associated with $\xi\in \mathfrak{g}$ as $$\xi_M(x)=\frac{\d}{\d t}\bigg|_{t=0}\Phi(\exp(t\xi),x)\,,\quad\forall x\in M\,.$$} corresponding to $\xi\in\mathfrak{g}$.
\end{definition}
Note that a contact Lie group action has a unique momentum map. The contact momentum map is Ad-equivariant, i.e. $J\circ \Phi(g,x)={\rm Ad}^*_{g^{-1}}J(x)$ for every $g\in G$ and for every $x\in M$ \cite{Gei2008}. The momentum map $J$ gives rise to a comomentum map $\lambda\colon\xi\in\mathfrak{g}\mapsto J_\xi\in\Cinfty(M)$ defined by $J_\xi(x) = \langle J(x),\xi\rangle$ for every $x\in M$ and $\xi\in\mathfrak{g}$.

\begin{proposition}(See {\rm\cite[Prop. 3.1]{Wil2002}} for a proof)
    Let $\Phi:G\times M\rightarrow M$ be a proper contact Lie group action relative to a contact manifold $(M,\eta)$. Consider its associated contact momentum map $J\colon M\to\mathfrak{g}^\ast$. Then,
    \begin{enumerate}[(1)]
        \item The level sets of the momentum map $J$ are invariant under the action of the flow of the Reeb vector field of $(M,\eta)$.
        \item For every $x\in M$, $v\in\T_x M$, and $\xi\in\mathfrak{g}$, one has
        $$ \d J_\xi = -i_{\tilde\xi}\d\eta\,. $$
        \item If $J(x) = 0$, we have that $\T(G\cdot x)$ is an isotropic subspace of the symplectic vector space $(\ker\eta_x, \d_x\eta|_{\ker \eta_x})$.
        \item $(\Ima\T_x J)^\circ = \{ \xi\in\mathfrak{g}\mid \tilde\xi_x\in\ker\d_x\eta \}$.
    \end{enumerate}
\end{proposition}

Note that the fundamental vector fields of a contact Lie group action have Hamiltonian functions that are first-integrals of the Reeb vector field. This fact is relevant to prove the following proposition.

\begin{proposition}
    Let $J:M\rightarrow \mathfrak{g}^*$ be a contact momentum map relative to $(M,\eta)$ for a contact Lie group action $\Phi:G\times M\to M$. Then, the mapping $\xi\in \mathfrak{g}\mapsto J_\xi\in \Cinfty_g(M)$ is a Lie algebra morphism. Moreover, 
    $$
        J:x\in M\longmapsto J(x)\in \mathfrak{g}^*
    $$
    induces a Poisson algebra morphism $J^*:f\in \Cinfty(\mathfrak{g}^*)\mapsto f\circ J\in \Cinfty_g(M)$ relative to the Kirillov--Kostant--Souriau bracket on $\mathfrak{g}^*$. 
\end{proposition} 
\begin{proof}
    Taking into account that $R J_\xi = 0$ for every $\xi\in\mathfrak{g}$, we have, for an arbitrary $\mu\in\mathfrak{g}$, that
    \begin{align*}
        i_{\tilde\xi}\d J_\mu &= -i_{\tilde\xi}\left( i_{\tilde\mu}\d\eta - (RJ_\mu)\eta \right) = \d\eta(\tilde\xi,\tilde\mu) = -\{J_\xi,J_\mu\} - J_\xi R J_\mu + J_\mu R J_\xi = -\{J_\xi,J_\mu\}\,.
    \end{align*}
    On the other hand, since $\Phi$ is Ad-equivariant, one has 
    $$ i_{\tilde\xi}\d J_\mu = \tilde\xi J_\mu = -\langle J,[\xi,\mu] \rangle = -J_{[\xi,\mu]}\,,\qquad\forall\xi,\mu\in\mathfrak{g}\,, $$
Which shows that $\xi\in \mathfrak{g}\mapsto J_\xi\in \Cinfty(M)$ is a Lie algebra morphism. 

Since the Kirillov--Kostant--Souriau bracket is a Poisson bracket and the space of good Hamiltonian functions relative to $(M,\eta)$ is a Poisson algebra relative to the bracket (\ref{Eq:ContactBrack}), for all functions $f,g\in \Cinfty(\mathfrak{g}^*)$ and a basis $\{e_1,\ldots,e_r\}$ of $\mathfrak{g}\simeq \mathfrak{g}^{**}$, it follows that
$$
\{f,g\}_{\mathfrak{g}^*}\circ J=\frac{\partial f}{\partial e_i}\circ J\frac{\partial g}{\partial e_j}\circ J\{e_i,e_j\}_{\mathfrak{g}^*}\circ J=
\frac{\partial f}{\partial e_i}\circ J\frac{\partial g}{\partial e_j}\circ Jc_{ijk}e_k\circ J=\frac{\partial f}{\partial e_i}\circ J\frac{\partial g}{\partial e_j}\circ J\{e_i\circ J,e_j\circ J\}=\{f\circ J,g\circ J\},
$$
for $[e_i,e_j]=c_{ijk}e_k$ and $i,j=1,\ldots,r$.
\end{proof}

\begin{definition}
    Let $\Phi:G\times M\rightarrow M$ be a proper contact Lie group action on a contact manifold $(M,\eta)$. Consider its associated contact momentum map $J\colon M\to\mathfrak{g}^\ast$ and $\mu\in\mathfrak{g}^\ast$. The \textit{kernel group} of $\mu$ is the unique connected Lie subgroup of $G_\mu\subset G$ with Lie algebra $\mathfrak{k}_\mu = \ker \restr{\mu}{\mathfrak{g}_\mu}$, where $\mathfrak{g}_\mu$ is the Lie algebra of the isotropy group $G_\mu$ of the point $\mu\in \mathfrak{g}^*$ relative to the coadjoint action of $G$ on $\mathfrak{g}$. We denote by $K_\mu$ the kernel group of $\mu$. The \textit{contact quotient}, or \textit{contact reduction} of $M$ by $G$ at $\mu$ is
    $$ M_\mu = J^{-1}(\R^+\mu) / K_\mu\,. $$
\end{definition}


\begin{theorem}\label{thm:reduction-willet}
    Let $G$ be a Lie group acting by contactomorphisms on a contact manifold $(M,\eta)$, and let $J:M\to\mathfrak{g}^\ast$ be its associated contact momentum map. Let $K_\mu$, with $\mu\in\mathfrak{g}^\ast$, be the connected Lie subgroup of $G_\mu$ with Lie algebra $\mathfrak{k}_\mu = \ker \restr{\mu}{\mathfrak{g}_\mu}$. If
    \begin{enumerate}[(i)]
        \item $K_\mu$ acts properly on $J^{-1}(\R^+\mu)$,
        \item $J$ is transverse (see {\rm\cite{AMR88}} for a definition) to $\R^+\mu$,
        \item $\ker\mu + \mathfrak{g}_\mu = \mathfrak{g}$,
    \end{enumerate}
    then the quotient $M_\mu = J^{-1}(\R^+\mu)/K_\mu$, if a manifold, is naturally a contact manifold, i.e.
    $$ \ker\eta\cap \T\left(J^{-1}(\R^+\mu)\right) $$
    gives rise to a contact distribution on the quotient $M_\mu$.
\end{theorem}

\section{Contact Lie systems}\label{Se:ConLieSys}

Let $V$ be a Lie algebra with Lie bracket $[\cdot,\cdot]\colon V\times V\to V$.  Given subsets $\A,\B\subset V$, we write $[\A,\B]$ for the real vector space generated by the Lie brackets between the elements of $\A$ and $\B$. Then, $\mathrm{Lie}(\A,V,[\cdot,\cdot])$, or simply $\mathrm{Lie}(\A)$, stands for the smallest Lie subalgebra of $V$ (in the sense of inclusion) containing $\A$.

A {\it $t$-dependent vector field} on $M$ is a map $X\colon\R\times M\to\T M$ such that, for every $t\in\R$, the map $X_t = X(t,\cdot)\colon M\to\T M$ is a vector field. In fact, a $t$-dependent vector field $X$ on $M$ amounts to a $t$-parametric family of vector fields $X_t$ on $M$ with $t\in\R$. An {\it integral curve} of $X$ is an integral curve, $\gamma\colon t\in\R\mapsto(t,x(t))\in\R\times M$, of the {\it autonomisation} of $X$, namely $\tparder{}{t} + X$ understood in the natural way as an element in $\X(\R\times M)$. Every $t$-dependent vector field, $X$, on $M$ gives rise to its referred to as {\it associated system} given by
\begin{equation}\label{eq:system-normal-form}
    \frac{\d x}{\d t} = X(t,x)\,,\qquad \forall t\in\R\,,\quad \forall x\in M\,.
\end{equation}
The curves $\gamma:t\in \mathbb{R}\mapsto (t,x(t))\in \mathbb{R}\times M$, where $x(t)$ is a solution of the above system of differential equations, are the {\it integral curves} of $X$. Conversely, every system of first-order differential equations in normal form in $M$, that is \eqref{eq:system-normal-form}, describes the integral curves of a unique $t$-dependent vector field $X$ on $M$. Hence, this allows us to identify $X$ with its associated system, namely \eqref{eq:system-normal-form}, and to use $X$ to refer to both. Such a notation will not lead to contradiction, as it will be clear from context what we mean by $X$ in each case. The {\it smallest Lie algebra} of a $t$-dependent vector field $X$ is the Lie algebra $V^X = \mathrm{Lie}(\{X_t\}_{t\in\R})$. Every Lie algebra of vector fields $V$ on $M$ gives rise to an associated distribution on $M$ of the form
$$
\mathcal{D}^V_x=\{X_x:X\in V\}\,,\qquad \forall x\in M\,.
$$
In particular, a $t$-dependent vector field $X$ on $M$ gives rise to an associated distribution, $\mathcal{D}^X$, given by $\mathcal{D}^{X}=\mathcal{D}^{V^X}$. It is worth noting that $\mathcal{D}^V$ does not need to have constant rank at every point of $M$, namely the subspaces $\mathcal{D}_x^V$ may have different dimensions for different points $x\in M$.

A {\it Lie system} is a $t$-dependent vector field $X$ on a manifold $M$ whose smallest Lie algebra $V^X$ is finite-dimensional \cite{LS2020}. If $X$ takes values in a finite-dimensional Lie algebra of vector field $V$, i.e. $\{X_t\}_{t\in \mathbb{R}}\subset 
 V$, we call $V$ a {\it Vessiot--Guldberg Lie algebra} of $X$ and it is said that $X$ admits a Vessiot--Guldberg Lie algebra. A $t$-dependent vector field $X$ admits a Vessiot--Guldberg Lie algebra if, and only if, $V^X$ is finite-dimensional. An \textit{automorphic Lie system} is a Lie system, $X^G$, on a Lie group $G$ admitting a Vessiot--Guldberg Lie algebra given by the space of right-invariant vector fields, $\X_R(G)$, on $G$. A \textit{locally automorphic Lie system} is a triple $(M,X,V)$ such that $V$ is a Vessiot--Guldberg Lie algebra of $X$ whose associated distribution, $\mathcal{D}^V$, is equal to $\T M$. 

 The main property of Lie systems is the so-called superposition rule \cite{CL2011,Win1983}. A {\it superposition rule} for a system $ X $ on $ M $ is a map $ \Phi : M^{k} \times M \to M$ such that the general solution $ x(t) $ of $ X $ can be written as $ x(t) = \Phi(x_{(1)}(t), \dots, x_{(k)}(t); \rho) $, where $ x_{(1)}(t), \dots, $ $ x_{(k)}(t) $ is a generic family of particular solutions and $ \rho $ is a point in $ M $ related to the initial conditions of $X$. The Lie Theorem \cite{CGM07,CL2011,Win1983} states that a system $X$ admits a superposition rule if and only if it is a Lie system.

A {\it Lie--Hamilton system} is a triple $(M,\Lambda,X)$, where $X$ is a Lie system on $M$ admitting a Vessiot--Guldberg Lie algebra of Hamiltonian vector fields relative to a Poisson bivector $\Lambda$ on $M$. If $\Lambda^\sharp$ is invertible, it gives rise to a symplectic form $\omega$ such that $\omega^\flat=\Lambda^\sharp$, and we will sometimes denote $(M,\Lambda,X)$ by $(M,\omega,X)$. Lie--Hamilton systems became relevant as they allowed the use of symplectic and Poisson techniques for the simple determination of superposition rules, Lie symmetries, constants of motion, and other properties of Lie--Hamilton systems \cite{LS2020}.  
Finally, a {\it Jacobi--Lie system} is a Lie system $X$ on $M$ admitting a Vessiot--Guldberg Lie algebra of Hamiltonian vector fields relative to a Jacobi manifold $(M,\Lambda,E)$. We call \textit{Jacobi--Lie Hamiltonian system} a quadruple $(M,\Lambda,E,h)$, where $(M,\Lambda,E)$ is a Jacobi manifold and $h:(t,x)\in \R\times M \mapsto h_t(x)\in N$ is a $t$-dependent function such that $\mathrm{Lie}(\{h_t\}_{t\in\R},\{\cdot,\cdot\})$ is a finite-dimensional Lie algebra relative to the Lie bracket $\{\cdot,\cdot\}$ associated with the Jacobi manifold $(M,\Lambda,E)$. Given a system $X$ on $M$, we say that $X$ admits a Jacobi--Lie Hamiltonian system $(M, \Lambda, E, h)$ if $X_t$ is a Hamiltonian vector field with Hamiltonian function $h_t$ (with respect to $(M, \Lambda, E)$) for each $t\in\R$ \cite{AHKR22,AHR22,HLS15,LS2020}. We hereafter write ${\rm Cas}(M,\Lambda,E)$ the space of Hamiltonian functions related to a zero vector field with respect to a Jacobi manifold $(M,\Lambda,E)$.

\begin{example}{\bf (Riccati equations)}
    Consider the differential equation
    \begin{equation}\label{eq:Riccati}
        \frac{\d x}{\d t} = a_1(t) + a_2(t)x + a_3(t)x^2\,,
    \end{equation}
    where $a_1(t),a_2(t),a_3(t)$ are arbitrary $t$-dependent functions. System \eqref{eq:Riccati} is the system associated with the $t$-dependent vector field
    $$ X(t,x) = \sum_{\alpha = 1}^ 3a_\alpha(t)X_\alpha(x)\,, $$
    where
    $$ X_1 = \parder{}{x}\,,\quad X_2 = x\parder{}{x}\,,\quad X_3 = x^2\parder{}{x}\,. $$
    Since 
    $$ [X_1,X_2] = X_1\,,\quad [X_1,X_3] = 2X_2\,,\quad [X_2,X_3] = X_3\,, $$
    it follows that $X_1,X_2,X_3$ span a Lie algebra isomorphic to $\slfr_2$. Thus, $X$ defines a Lie system on $\R$ with Vessiot--Guldberg Lie algebra $\langle X_1,X_2,X_3\rangle\simeq\slfr_2$.
\end{example}

\begin{definition}
    A {\it contact Lie system} is a triple $(M,\eta,X)$, where $\eta$ is a contact form on $M$ and $X$ is a Lie system on $M$ whose smallest Lie algebra $V^X$ is a finite-dimensional real Lie algebra of contact Hamiltonian vector fields relative to $\eta$. A {\it contact Lie system} is called \textit{conservative} if the Hamiltonian functions of the vector fields in $V^X$ are first-integrals of the Reeb vector field of $(M,\eta)$.
\end{definition}

Note that a conservative contact Lie system amounts to a contact Lie system $X$ on a manifold $M$ relative to a contact manifold $(M,\eta)$ that is invariant relative to the flow of the Reeb vector field, $R$, of $\eta$, namely $R$ is a Lie symmetry of $X$.

A Lie system $X$ can be considered as a curve in $V^X$. In contact manifolds, every Hamiltonian vector field gives rise to a unique Hamiltonian function. Therefore, $V^X$ gives rise to a linear space of functions $\mathfrak{W}$ and $X$ defines a curve in $\mathfrak{W}$. Due to the isomorphism of Lie algebras between the space of Hamiltonian vector fields of $(M,\eta)$ and $\Cinfty(M)$, it turns out that $\mathfrak{W}$ is a Lie algebra. This suggests us the following definition.

\begin{definition} A {\it contact Lie--Hamiltonian} is a triple $(M,\eta,h:\mathbb{R}\times M\rightarrow \mathbb{R})$, where $(M,\eta)$ is a contact manifold and $h$ gives rise to a $t$-dependent family of functions $h_t:x\in M
\mapsto h(t,x)\in \mathbb{R}$, with $t\in \mathbb{R}$, that span a finite-dimensional Lie algebra of functions relative to the bracket in $\Cinfty(M)$ induced by $(M,\eta)$.
\end{definition}

Note that every contact Lie system gives rise a unique contact Lie--Hamiltonian and conversely.

\begin{example}{\bf (A simple control system)}\label{ex:simple-control}
Consider the system of differential equations in $\R^3$ given by
 \begin{equation}\label{eq:simple-control}
    \begin{dcases}
        \frac{\d x}{\d t} = b_1(t)\,,\\
        \frac{\d y}{\d t} = b_2(t)\,,\\
        \frac{\d z}{\d t} = b_2(t)x\,,
     \end{dcases}
 \end{equation}
where $b_1(t),b_2(t)$ are two arbitrary functions depending only on time. The relevance of this system is due to its occurrence in control problems {\rm\cite{Ra11}}.

System \eqref{eq:simple-control} describes the integral curves of the $t$-dependent vector field on $\R^3$ given by
\begin{equation}\label{eq:t-field-simple-control}
    X = b_1(t)X_1 + b_2(t)X_2\,,
\end{equation}
where
$$ X_1 = \parder{}{x}\,,\quad X_2 = \parder{}{y} + x\parder{}{z}\,. $$
The vector fields $X_1, X_2$, along with the vector field $X_3 = \partial/ \partial {z}$, span a three-dimensional Vessiot--Guldberg Lie algebra $V = \langle X_1,X_2,X_3\rangle\simeq\mathfrak{h}_3$ of $X$, where $\mathfrak{h}_3$ is the so-called three-dimensional Heisenberg Lie algebra. Indeed, the commutations relations for $X_1,X_2,X_3$ read
$$ [X_1,X_2] = X_3\,,\quad [X_1,X_3] = 0\,,\quad [X_2,X_3] = 0\,. $$
The vector fields $X_1,X_2,X_3$ are contact Hamiltonian vector fields with respect to the contact form on $\mathbb{R}^3$ given by
$$ \eta_c = \d z - y\,\d x\,, $$
with Hamiltonian functions
$$ h_1 = y\,,\quad h_2 = -x\,,\quad h_3 = -1\,, $$
respectively. It follows that all the elements of $V^X$ are Hamiltonian vector fields relative to $(\R^3,\eta_c)$. Hence, the $t$-dependent Hamiltonian for \eqref{eq:simple-control} relative to $(\mathbb{R}^3,\eta_c)$ is given by
$$
h(t)=b_1(t)y-b_2(t)x\,.
$$
Thus, $(\R^3,\eta,X)$ is a contact Lie system. Since $h_1,h_2,h_3$ are first-integrals of $X_3=\partial/\partial z$, which is the Reeb vector field of $\eta_c$, then $(\R^3,\eta_c,X)$ is conservative. In fact, $[X_3,X_t]=0$ for every $t\in \mathbb{R}$. 

Note that $\eta_c$ gives rise to a volume form $\Omega_{\eta_c}=\eta_c\wedge \d\eta_c$ on $\mathbb{R}^3$. Since $h_1,h_2,h_3$ are first-integrals of the Reeb vector field of $\eta_c$,  the evolution of \eqref{eq:simple-control} leaves $\Omega_{\eta_c}$ invariant.
\end{example}

\bigskip

Let us study the behaviour of the volume form, $\Omega_\eta=\eta\wedge (\d\eta)^n$, induced by a $(2n+1)$-dimensional contact  manifold $(M,\eta)$ relative to the dynamics of a contact Lie system on $M$ relative to $(M,\eta)$.

\begin{proposition}\label{Prop:ConLiou} Let $(M,\eta,X)$ be a conservative contact Lie system on a $(2n+1)$-dimensional contact manifold $(M,\eta)$ and let $\Omega_\eta=\eta\wedge (\d\eta)^n$, then 
$$
\Lie_{X_t}\Omega_\eta= 0 \,,\qquad \forall t\in \R\,.
$$
\end{proposition}
\begin{proof}
 Recall that the vector fields of the smallest Lie algebra of a conservative contact Lie system are of the form $X_f$ for $Rf=0$ and a certain $f\in \Cinfty(M)$. Then,
$$
\Lie_{X_f}\Omega_\eta=\Lie_{X_f}(\eta\wedge (\d\eta)^n)=(\Lie_{X_f}\eta)\wedge (\d\eta)^n+n\eta\wedge \d\Lie_{X_f}\eta\wedge (\d\eta)^{n-1}=-(n+1)(Rf)\Omega_{\eta}\,,
$$
since $\Lie_{X_f}\eta=-(Rf)\eta$. As $Rf=0$, the result follows.
\end{proof}
It is worth noting that, given a contact manifold $(M,\eta)$, the space of Hamiltonian vector fields on $M$ admitting a Hamiltonian function being a first-integral of $R$ is a Lie subalgebra of $\Xh(M)$.


%


\begin{theorem} {\bf (Gromov’s non-squeezing theorem)}. Let $(M,\omega)$ be a symplectic manifold and let $\{q^1,\ldots,q^n,p_1,\ldots,p_n\}$ be Darboux coordinates on an open subset $U\subset M$. Given the set of points
$$
B(r)=\left\{(q,p)\in U:\sum_{i=1}^n\left[(q^i-q^i_0)^2+(p_i-p_i^0)^2\right]\leq r^2\right\},
$$
where $(q_0^1,\ldots,q_0^n,p_1^0,\ldots,p_n^0)\in U$,
if the image of $B(r)$ under a symplectomorphism $\phi:M\rightarrow M$ is such that $\phi(B(r))\subset C_R$, where
$$
C_R=\left\{(q,p)\in U:(q^1-q^1_0)^2+(p_1^0-p_1^0)^2\leq R^2\right\},
$$
then $r\geq R$.
\end{theorem}

The interest of the Gromov's non-squeezing theorem is due to the fact that it applies to the Hamiltonian system relative to a symplectic form appearing as the projection of a conservative contact Lie system $(M,\eta,X)$ onto the space of integral submanifolds of $R$ in $M$, i.e. $M/R$, if the latter admits a manifold structure \cite{AMR88}.

\subsection{Contact Lie systems and other classes of Lie systems}\label{subsec:other}

Recall that {Lie--Hamilton systems} are the Lie systems admitting a Vessiot--Guldberg Lie algebra of Hamiltonian vector fields relative to a Poisson bivector. They were the first studied type of Lie systems admitting a Vessiot--Guldberg Lie algebra of Hamiltonian vector fields relative to a geometric structure \cite{Car2000,CLS13}. Despite that, they were insufficient for studying many types of Lie systems \cite{LS2020}. Let us study why contact Lie systems are interesting on their own and their relations to other types of Lie systems. Let us start by the next proposition, which is a no-go result for the existence of a Poisson structure turning the vector fields of a Vessiot--Guldberg Lie algebra of a Lie system into Hamiltonian vector fields. It is indeed a version of a proposition in \cite{Car2014}.

\begin{proposition}\label{Pr:NoGo}
    If $X$ is a Lie system on an odd-dimensional manifold $M$ such that $\mathcal{D}^X=TM$, then $X$ does not give rise to any Lie--Hamilton system $(M,\Lambda,X)$ relative to any Poisson bivector $\Lambda$.
\end{proposition}
\begin{proof}
    Let us prove the proposition by reductio ad absurdum. The characteristic distribution of a Poisson bivector on a manifold is a distribution whose rank is even, but not necessarily constant, at every point of the manifold \cite{Va94}. Hence, all Hamiltonian vector fields must take values in a distribution that must have even rank at every point. Meanwhile, the vector fields of the smallest Lie algebra of $X$ span, by assumption, a distribution of odd-rank. Since all the vector fields of the smallest Lie algebra of $X$ are Hamiltonian by assumption, the unique distribution where they can take values in has odd rank. But then, they cannot be contained in a characteristic distribution of even rank at every point. This is a contradiction and $X$ does not give rise to a Lie--Hamilton system relative to any Poisson structure.
\end{proof}

Proposition \ref{Pr:NoGo} shows that Lie--Hamilton systems are not appropriate to describe Lie systems admitting certain smallest Lie algebras. Note that, for instance,  Example \ref{ex:simple-control} describes a Lie system whose smallest Lie algebra satisfies the conditions of Proposition \ref{Pr:NoGo} when the vectors $(b_1(t),b_2(t))$, with $t\in \mathbb{R}$, span $\mathbb{R}^2$ and, therefore, $V^X=\langle X_1,X_2,X_3\rangle$ while $\mathcal{D}^X=\T\mathbb{R}^3$. This illustrates the need for describing Lie systems admitting Vessiot--Guldberg Lie algebras of Hamiltonian vector fields relative to other geometric structures, like contact manifolds.

The following proposition shows how conservative contact Lie systems induce some Lie--Hamilton systems on other spaces.

\begin{proposition}\label{prop:conservative-contact-project-symplectic} If $(M,\eta,X)$ is a conservative contact Lie system and the space of integral curves of the Reeb vector field $R$, let us say $M/R$, is a manifold and $\pi_R:M\rightarrow M/R$ is the canonical projection, then $(M/R,\Omega,\pi_*X)$, where $\pi_R^*\Omega=\d\eta$, is a Lie--Hamilton system relative to the symplectic form $\Omega$ on $M/R$.
\end{proposition}
\begin{proof} Since $(M,\eta,X)$ is conservative, the Lie derivative of the Reeb vector field $R$ with Hamiltonian vector fields, e.g. the elements of $V^X$, is zero. Therefore, all the elements of $V^X$ are projectable onto $M/R$. Moreover, $\Lie_R\d\eta=0$ and $i_R\d\eta=0$. Hence, $\d\eta$ can be projected onto $M/R$. In other words, there exists a unique two-form, $\Omega$, on $M/R$ such that $\pi^*\Omega=\d\eta$. Note that $\Omega$ will is closed. Moreover, if $i_{Y_{[x]}}\Omega_{[x]}=0$ for a tangent vector $Y_{[x]}$ at a point $[x]$ in $N/R$, then there exists a tangent vector $\widetilde{Y}_x\in\T_xM$ projecting onto $Y_{[x]}\in\T_{[x]}M/R$. Then, $\pi^*i_{Y_{[x]}}\Omega_{[x]}=i_{\widetilde{Y}_x}(\d\eta)_x=0$. Hence, $\widetilde{Y}_x$ takes values in the kernel of $(\d\eta)_x$ and it is proportional to $R_x$. Hence, $\pi_{*x}Y_x=0$ and $\Omega$ is non-degenerate. Since $\Omega$ is closed, it becomes a symplectic form and the vector fields of $\pi_*V^X$ span a finite-dimensional Lie algebra of Hamiltonian vector fields relative to $\Omega$. Therefore, the $t$-dependent vector field $\pi_*X$, namely the $t$-dependent vector field $(\pi_*X)_t=\pi_*X_t$ for every $t\in \mathbb{R}$, becomes a Lie--Hamilton system relative to $\Omega$. 
\end{proof}

Since Hamiltonian vector fields relative to a contact Lie system are Hamiltonian vector fields relative to its associated Jacobi manifold, one may ask whether contact Lie systems are interesting on its own. There are several reasons. For instance, contact structures have particular properties that are not shared by general Jacobi manifolds and they are specific. For example, every Hamiltonian function determines a unique Hamiltonian vector field and conversely, which make some results more specific, e.g. every contact Lie system admits a contact Lie--Hamiltonian.


\begin{proposition}\label{Eq:ExtHam}
Every contact Lie system $(M,\eta,X)$ gives rise to a Lie--Hamilton system  $(\mathbb{R}\times M,e^{-s}(\d\eta + \eta\wedge\d s),\partial/\partial s+X)$, where $s$ is the natural variable in $\mathbb{R}$.
\end{proposition}

Proposition \ref{Eq:ExtHam} may be inappropriate to study contact Hamiltonian systems on $M$ via Hamiltonian systems on symplectic manifolds since the dynamics of a contact Hamiltonian vector field on $M$ may significantly differ from the Hamiltonian system on $\mathbb{R}\times M$ used to study it. For example, a contact Hamiltonian vector field $X$ on $M$ may have stable points, while $\partial/\partial s +X$, which is its associated Hamiltonian vector field on $\mathbb{R}\times M$, has not. This has relevance in certain theories, like the energy-momentum method \cite{Sim2020}. Moreover, every contact Lie system can be understood as the projection of a Lie--Hamilton system on a homogeneous symplectic manifold (see \cite{GG22}). Anyhow, the latter approach is not  appropriate for our purposes for a number of reasons, e.g. considering Lie systems on manifolds of larger dimension may make the study of the contact Lie system harder to solve. Examples of this problem will be given in the next section.  






Finally, let us recall that a multisymplectic Lie system is triple $(M,\Omega,X)$, where $X$ is a Lie system on $M$ admitting a Vessiot--Guldberg Lie algebra of Hamiltonian vector fields relative to the multisymplectic form $\Omega$ on $M$ (see \cite{GLMV19,GLRRV22} for details). The following proposition, whose proof is immediate, relates conservative contact Lie systems to multisymplectic Lie systems.

\begin{corollary} If $(M,\eta,X)$ is a conservative contact Lie system, then $(M,\Omega_\eta,X)$ is a multisymplectic Lie system.
\end{corollary}

\section{Existence of contact forms for Lie systems}\label{Se:ExiCon}

Let us analyse the existence of contact forms turning the vector fields of a Vessiot--Guldberg Lie algebra into Hamiltonian vector fields. Our results will help us to determine Lie systems that can be considered as contact Lie systems. In particular, the classification of automorphic Lie systems on three-dimensional Lie groups admitting a left-invariant contact form will be given.

\begin{lemma}\label{Lemma:Uni}
    Let $X$ be a Lie system on a manifold $M$ with a smallest Lie algebra $V^X$ such that $\mathcal{D}^{V^X}=\T M$. If $\eta$ is a differential form on $M$ such that $\Lie_X\eta=0$ for every $X\in V^X$, then the value of $\eta$ at a point $M$ determines the value of $\eta$ on the whole $M$.
\end{lemma}
\begin{proof} Let $x\in M$ be a fixed arbitrary point. Since the vector fields in $V^X$ span the distribution $\T M$, it follows from basic control theory \cite{BR84} that $x$ can be connected to any other point $y\in M$ by a local diffeomorphism of the form
\begin{equation}\label{eq:phi_xy}
    \phi_{xy}=\exp(t_1X_{i_1})\circ \exp(t_2X_{i_2})\circ\cdots\circ \exp(t_kX_{i_k})\,, 
\end{equation}
where $k\in \mathbb{N}$ is a natural number or zero, $i_1,\ldots,i_k\in\{1,\ldots,r\}$, the vector fields $ X_1,\ldots,X_r$ form a basis of $V^X$, and $t_1,\ldots,t_k\in \mathbb{R}$. Since $\Lie_X\eta=0$ for every $X\in V$ and due to \eqref{eq:phi_xy}, it follows that $\phi_{xy}^*\eta_y=\eta_x$ and the value of $\eta_y$ is determined by $\eta_x$.
\end{proof}

\begin{proposition}\label{Pr:ClasInv} Given a locally automorphic Lie system $(M,V^X,X)$, there exists a bijection between the space $\mathcal{C}$ of contact forms turning the elements of $V^X$ into Hamiltonian vector fields and the one-chains, $\vartheta$, of the Chevalley--Eilenberg cohomology of $\mathfrak{g}$ isomorphic to $V^X$ such that $\vartheta\wedge (\delta\vartheta)^k$ is a non-zero $(2k+1)$-covector with $\dim M=2k+1$.
\end{proposition}
\begin{proof} By Lemma \ref{Lemma:Uni} and our assumptions, a contact form on $M$ is determined by its value at one point $x\in M$. Every locally automorphic Lie system $(M,V^X,X)$ is locally diffeomorphic to an automorphic Lie system \cite{GLMV19}, namely, in our case, a Lie system on a Lie group $G$ with Lie algebra $\mathfrak{g}$ so that
\begin{equation}\label{eq:aut}
    \frac{\d g}{\d t}=\sum_{\alpha=1}^r b_\alpha(t)X_\alpha^R(g)\,,\qquad \forall g\in G\,,
\end{equation}
for a basis of right-invariant vector fields $X^R_1,\dotsc,X^R_r$ on $G$ and some functions $b_1(t),\ldots,b_r(t)$. Since $V^X$ is the smallest Lie algebra containing the vector fields $\{X_t\}_{t\in \mathbb{R}}$, as $\mathcal{D}^{V^X} = \T N$, and \eqref{eq:aut} is locally diffeomorphic to $X$, it follows that the smallest Lie algebra of \eqref{eq:aut} is $\langle X_1^R,\dotsc,X_r^R\rangle$, which spans $\T G$. The local diffeomorphism maps the invariant contact form for $X$ to a left-invariant contact form $\eta^L$ for \eqref{eq:aut}. As $\eta^L$ is a left-invariant contact form, then $\eta^L\wedge (\d\eta^L)^k$ is a volume form on $G$ for $2k+1=\dim G=\dim V^X=\dim M$. Moreover,
$$
\d\eta^L(X^L_i,X^L_j)=-\eta^L([X^L_i,X^L_j])\,,\qquad i,j=1,\ldots,r\,.
$$
Define $\delta:\mathfrak{g}^*\rightarrow \bigwedge^2\mathfrak{g}^*$ to be minus the transpose of $[\cdot,\cdot]:\bigwedge^2\mathfrak{g}\rightarrow \mathfrak{g}$. On the other hand, $\eta^L\wedge(\d\eta^L)^k$ being a volume form amounts to the fact that its value at the neutral element $e$ is different from zero. But $\eta^L_e\wedge(\d\eta^L)^k_e = \eta_e^L\wedge(\delta\eta_e^L)^k$.
\end{proof}

Note that the conditions in Proposition \ref{Pr:ClasInv} can be checked for every automorphic Lie system on a three-dimensional Lie group with a smallest Lie algebra given by the right-invariant vector fields on the Lie group, as their Lie algebras are completely classified. It was proved in \cite{FJ15,LW20} that every real non-abelian three-dimensional Lie algebra is isomorphic to $(E,[\cdot,\cdot])$, where $E$ is a three-dimensional vector space and the Lie bracket is given on a canonical basis $\{e_1,e_2,e_3\}$ of $E$ by one of the cases in Figure \ref{fig:Lie-algebra-classification}. Note that it is not appropriate in our classification to relate contact Lie systems on three-dimensional Lie groups to Hamiltonian Lie systems on four-dimensional manifolds for evident reasons, e.g. this approach just makes the problem much harder to solve as it demands to analyse a problem on a four-dimensional Lie group and to study how the latter is related to the solution of our initial problem. 


Let us now classify left-invariant contact forms for automorphic Lie systems on  three particular types of three-dimensional Lie groups, namely those with Lie algebras $\mathfrak{sl}_2$, $\mathfrak{r}_{3,\lambda}$, and $\mathfrak{r}'_{3,\lambda\neq 0}$. More specifically, we will study the conditions required for an element $\lambda_1 e^1 + \lambda_2 e^2 + \lambda_3 e^3$, where $\{e^1,e^2,e^3\}$ is the dual basis to the basis $\{e_1,e_2,e_3\}$ of $T_eG$ and $\lambda_1,\lambda_2,\lambda_3\in \mathbb{R}$, to be the value of a left-invariant contact form on a three-dimensional Lie group at the neutral element.

$\bullet$ Case $\mathfrak{sl}_2$: 
The corresponding Lie bracket is an antisymmetric bilinear function that can be understood univocally as a mapping $[\cdot,\cdot]:\mathfrak{sl}_2\wedge \mathfrak{sl}_2\rightarrow \mathfrak{sl}_2$. Defining the map $\delta:\mathfrak{sl}_2^*\rightarrow \mathfrak{sl}_2^*\wedge\mathfrak{sl}_2^*$ as $\delta = -[\cdot,\cdot]^T$, we have
$$
    \delta(e^1) = -e^1([\cdot,\cdot])=   \frac{1}{2}e^3\wedge e^2\,,\quad
    \delta(e^2) = -e^2([\cdot,\cdot]) = -\frac{1}{2}e^1\wedge e^2\,,\quad
    \delta(e^3) = -e^3([\cdot,\cdot]) = \frac{1}{2}e^1\wedge e^3\,,
$$
and thus,
$$ \delta = \frac{1}{2}e_1\otimes e^3\wedge e^2 - \frac{1}{2}e_2\otimes e^1\wedge e^2 + \frac{1}{2}e_3\otimes e^1\wedge e^3\,. $$
In this case, $k=1$ and
\begin{multline*}
    0 \neq \delta(\lambda_1e^1 + \lambda_2e^2 + \lambda_3e^3)\wedge(\lambda_1e^1 + \lambda_2e^2 + \lambda_3e^3)\\= \frac 12\left( \lambda_1 e^3\wedge e^2 - \lambda_2e^1\wedge e^2 +  \lambda_3 e^1\wedge e^3 \right)\wedge(\lambda_1e^1 + \lambda_2e^2 + \lambda_3e^3)
    \\= -\frac{1}{2}\left(\lambda_1^2 + 2\lambda_2\lambda_3\right) e^1\wedge e^2\wedge e^3\,.
\end{multline*}
Then, the differential one-form $\eta^L=\sum_{\alpha=1}^3\lambda_\alpha\eta^L_\alpha$ on $\SL(2,\mathbb{R})$, where $\eta^L_\alpha(e)=e_\alpha$ for $\alpha=1,2,3$, is a contact form if and only if $\lambda_1^2 + 2\lambda_2\lambda_3\neq 0$.
    
$\bullet$ Case $\mathfrak{r}_{3,\lambda}$, with  $\lambda\in (-1,1)$. As previously, define the map $\delta:\mathfrak{r}_{3,\lambda}^*\rightarrow \mathfrak{r}_{3,\lambda}^*\wedge\mathfrak{r}_{3,\lambda}^*$ as $\delta = -[\cdot,\cdot]^T$. Then,
$$
    \delta(e^1) = -e^1([\cdot,\cdot])=   \frac{1}{2}e^1\wedge e^3\,,\quad
    \delta(e^2) = -e^2([\cdot,\cdot]) = -\frac{\lambda}{2}e^3\wedge e^2\,,\quad
    \delta(e^3) = -e^3([\cdot,\cdot]) = 0\,,
$$
and thus,
$$
    \delta = \frac{1}{2}e_1\otimes e^1\wedge e^3 - \frac{1}{2}\lambda e_2\otimes e^3\wedge e^2\,.
$$
Therefore,
\begin{multline*}
    0 \neq \delta(\lambda_1e^1 + \lambda_2\lambda e^2 + \lambda_3e^3)\wedge(\lambda_1e^1 + \lambda_2e^2 + \lambda_3e^3)\\= \left( \frac{\lambda_1}{2} e^1\wedge e^3 - \frac{\lambda_2\lambda}{2}e^3\wedge e^2\right)\wedge(\lambda_1e^1 + \lambda_2e^2 + \lambda_3e^3)\\
    = \frac{1}{2}\lambda_1\lambda_2(1-\lambda)e^1\wedge e^2\wedge e^3\,.
\end{multline*}
Then, the left-invariant contact forms on a Lie group with Lie algebra isomorphic to $\mathfrak{r}_{3,\lambda}$  are characterised by the condition $\lambda_1\lambda_2\neq 0$.

$\bullet$ Case $\mathfrak{r}'_{3,\lambda\neq 0}$. Defining the map $\delta:\mathfrak{r}_{3,\lambda\neq 0}^{\prime \,*}\rightarrow \mathfrak{r}_{3,\lambda\neq 0}^{\prime \,*}\wedge\mathfrak{r}_{3,\lambda\neq 0}^{\prime \,*}$ as $\delta = -[\cdot,\cdot]^T$, we have
$$
    \delta(e^1) = \frac{\lambda}{2}e^1\wedge e^3 - \frac{1}{2}e^3\wedge e^2\,,\quad
    \delta(e^2) = -\frac{1}{2}e^1\wedge e^3 - \frac{\lambda}{2}e^3\wedge e^2\,,\quad
    \delta(e^3) = 0\,,
$$
and thus,
$$
    \delta = \frac{\lambda}{2} e_1\otimes e^1\wedge e^3 - \frac{1}{2}e_1\otimes e^3\wedge e^2 - \frac{1}{2}e_2\otimes e^1\wedge e^3 - \frac{\lambda}{2}e_2\otimes e^3\wedge e^2\,.
$$
In this case,
\begin{multline*}
    0 \neq \delta(\lambda_1e^1 + \lambda_2e^2 + \lambda_3e^3)\wedge(\lambda_1e^1 + \lambda_2e^2 + \lambda_3e^3)\\
    = \left( \frac{\lambda\lambda_1}{2} e^1\wedge e^3 - \frac{\lambda_1}{2}e^3\wedge e^2 - \frac{\lambda_2}{2}e^1\wedge e^3 - \frac{\lambda\lambda_2}2e^3\wedge e^2 \right)\wedge(\lambda_1e^1 + \lambda_2e^2 + \lambda_3e^3)\\
    = \frac{1}{2}\left( \lambda_1^2 + \lambda_2^2\right) e^1\wedge e^2\wedge e^3\,.
\end{multline*}
Then, the differential one-form $\eta^L = \sum_{\alpha=1}^3\lambda_\alpha\eta^L_\alpha$ on each Lie group with Lie algebra $\mathfrak{r}_{3,\lambda\neq 0}^{\prime }$, where $\eta^L_\alpha(e)=e_\alpha$, with $\alpha=1,2,3$, is a contact form if and only if $\lambda_1^2 + \lambda_2^2 > 0$.

\bigskip

The other cases can be computed similarly, as summarised in the following theorem.

\begin{theorem}
    Let $G$ be a Lie group with a three-dimensional non-abelian Lie algebra $\mathfrak{g}$. Then, the left-invariant one-form $\eta^L = \sum_{\alpha=1}^3\lambda_\alpha\eta^L_\alpha$ on $G$, where $\eta^L_\alpha(e)=e^\alpha$ for $\alpha=1,2,3$ and $\lambda_i\in\R$, is a contact form if and only if the condition for the value of $\eta^L(e)$ in Figure \ref{fig:Lie-algebra-classification} for the Lie algebra $\mathfrak{g}$ of $G$ is satisfied.
    
\begin{figure}[ht]
    \centering
    \begin{tabular}{|c|c|c|c|c|}
        \hline
        Lie algebra & $[e_1,e_2]$ & $[e_1,e_3]$ & $[e_3,e_2]$ & Contact condition \\
        \hline\hline
        $\mathfrak{sl}_2$ & $e_2$ & $-e_3$ & $-e_1$ & $\lambda_1^2 + 2\lambda_2\lambda_3 > 0$ \\
        \hline
        $\mathfrak{su}_2$ & $e_3$ & $-e_2$ & $-e_1$ & $\lambda_1^2 + \lambda_2^2 + \lambda_3^2 > 0$ \\
        \hline
        $\mathfrak{h}_3$ & $e_3$ & $0$ & $0$ & $\lambda_3 \neq 0$ \\
        \hline
        $\mathfrak{r}'_{3,0}$ & $-e_3$ & $e_2$ & $0$ & $\lambda_2^2 + \lambda_3^2 > 0$ \\
        \hline
        $\mathfrak{r}_{3,-1}$ & $e_2$ & $-e_3$ & $0$ & $\lambda_2\lambda_3 \neq 0$ \\
        \hline
        $\mathfrak{r}_{3,1}$ & $e_2$ & $e_3$ & $0$ & $\nexists$ \\
        \hline
        $\mathfrak{r}_{3}$ & $0$ & $-e_1$ & $e_1 + e_2$ & $\lambda_1 \neq 0$ \\
        \hline
        $\mathfrak{r}_{3,\lambda}$ & $0$ & $-e_1$ & $\lambda e_2$ & $ \lambda_1\lambda_2\neq 0$ \\
        \hline
        $\mathfrak{r}'_{3,\lambda\neq 0}$ & $0$ & $e_2-\lambda e_1$ & $\lambda e_2 + e_1$ & $\lambda_1^2 + \lambda_2^2 > 0$ \\
        \hline
    \end{tabular}
    \caption{Classification of left-invariant contact forms on   non-abelian three-dimensional Lie algebras and left-invariant contact forms on their associated Lie groups. Note that $\lambda\in (-1,1).$}
    \label{fig:Lie-algebra-classification}
\end{figure}
\end{theorem}

The following proposition takes a deeper look at the properties of left-invariant contact forms on Lie groups and show some of their properties. In particular, it shows that the space of left-invariant contact forms on a Lie group must be invariant under the natural action of ${\rm Aut}(G)$, namely the space of Lie group automorphism of $G$, on $\mathfrak{g}^*$. Recall that $\mathrm{Aut}(G)$ acts on $G$, which gives rise to a Lie group action $(f_g,v)\in \mathrm{Aut}(G)\times\mathfrak{g}\mapsto \T_e f_g(v)\in\mathfrak{g}$ and its dual one.

\begin{proposition}
    Let ${\rm Aut}(G)$ be the Lie group of Lie group automorphisms of $G$ and let $\varphi:{\rm Aut}(G)\times \mathfrak{g}\rightarrow \mathfrak{g}$ be its associated action. Then, the space $\mathcal{C}$ of left-invariant contact forms on $G$ is invariant relative to the action of ${\rm Aut}(G)$ on $\mathfrak{g}^*$.
\end{proposition}
\begin{proof} Let us prove that every ${\rm Ad}_g$, with $g\in G$, maps left-invariant vector fields on $G$ into left-invariant vector fields on $G$. Given  a left-invariant vector field $X^L$ on $G$ with $X^L(e)=\xi$, one has
$$
X^L(g)=\frac{\d}{\d t}\bigg|_{t=0}g\exp(t\xi)\,,\qquad \forall g\in G\,,
$$
and then
$$
{\rm Ad}_{h*g}[X^L(g)]=\frac{\d}{\d t}\bigg|_{t=0}hg\exp(t\xi)h^{-1}=\frac{\d}{\d t}\bigg|_{t=0}hgh^{-1}\exp(t{\rm Ad}_{h*e}(\xi))=Y^L(hgh^{-1})\,,
$$
for the left-invariant vector field $Y^L$ on $G$ such that $ Y^L(e)={\rm Ad}_{g*e}(\xi)$. As left-invariant one-forms are dual to left-invariant vector fields,  ${\rm Ad}_{g*}\eta^L$ is a left-invariant one-form on $G$ for every $g\in G$. Hence, if $\eta^L$ is a left-invariant contact form on $G$ and $\dim G=2k+1$, one has that
$$
0\neq {\rm Ad}_{g*}[(\d\eta^L)^k\wedge \eta^L]=[\d {\rm Ad}_{g*}\eta^L]^k\wedge {\rm Ad}_{g*}\eta^L\,,\qquad \forall g\in G\,.
$$
And ${\rm Ad}_{g*}\eta^L$ is a new contact form.  Moreover, the value of $\eta_L(e)$ at the neutral element $e$ of $G$ is such that $[{\rm Ad}_g^*\eta^L]_e={\rm Ad}_{g*e}^T[\eta^L_e]$. Hence, if an element of $\mu\in \mathfrak{g}^*$ determines the value at $e$ of a left-invariant contact form, all left-invariant one-forms with values at $e$ within the coadjoint orbit of $\mu$ in $\mathfrak{g}^*$ give rise to contact forms.
\end{proof}

It is worth noting that since the tangent map at $e$ in $G$ to every element of $\mathrm{Aut}(G)$ is an element of $\mathrm{Aut}(\mathfrak{g})$ and vice versa, where ${\rm Aut}(\mathfrak{g})$ stands for the group of Lie algebra automorphisms of $\mathfrak{g}$, the action of $\mathrm{Aut}(G)$ on $\mathfrak{g}^*$ is indeed the action of $\mathrm{Aut}(\mathfrak{g})$ on $\mathfrak{g}^*$.




\section{Examples}\label{Se:Examples}


\subsection{The Brockett control system}
Let us consider a second example of contact Lie system. The Brockett control system \cite{Ra11} in $\R^3$ is given by
\begin{equation}\label{eq:Brocket-system}
\begin{dcases}
        \frac{\d x}{\d t} = b_1(t)\,,\\
        \frac{\d y}{\d t} = b_2(t)\,,\\
        \frac{\d z}{\d t} = b_2(t)x - b_1(t)y\,,
    \end{dcases}
  \end{equation}
where $b_1(t)$ and $b_2(t)$ are arbitrary $t$-dependent functions. System \eqref{eq:Brocket-system} is associated with the $t$-dependent vector field
$$ X = b_1(t)X_1 + b_2(t)X_2\,, $$
where
$$ X_1 = \parder{}{x} - y\parder{}{z}\,,\quad X_2 = \parder{}{y} + x\parder{}{z}\,, $$
along with the vector field $X_3 = 2\dparder{}{z}$, span a three-dimensional Vessiot--Guldberg Lie algebra $V = \langle X_1,X_2,X_3\rangle$ with commutation relations
$$ [X_1,X_2] = X_3\,,\quad [X_1,X_3] = 0\,,\quad [X_2,X_3] = 0\,. $$
As in Example \ref{ex:simple-control}, the vector space $\langle X_1,X_2,X_3\rangle$ is a Vessiot--Guldberg Lie algebra isomorphic to the three-dimensional Heisenberg Lie algebra $\mathfrak{h}_3$ (see Figure \ref{fig:Lie-algebra-classification}).

The Lie algebra of Lie symmetries of $V$,  i.e. the vector fields on $\mathbb{R}^5$ commuting with all the elements of $V^Q$, is spanned by the vector fields
$$ Y_1 = \parder{}{x} + y\parder{}{z}\,,\quad Y_2 = \parder{}{y} - x\parder{}{z}\,,\quad Y_3 = 2\parder{}{z}\,, $$
which have commutation relations
$$ [Y_1,Y_2] = -Y_3\,,\quad [Y_1,Y_3] = 0\,,\quad [Y_2,Y_3] = 0\, $$
Let us denote the Lie algebra of Lie symmetries of $V$ by $\Sym(V)$. The dual base of one-forms to $Y_1,Y_2,Y_3$ is
$$ \eta_1 = \d x\,,\quad \eta_2 = \d y\,,\quad \eta_3 = \frac{1}{2}(\d z - y\d x + x\d y)\,. $$
It is clear that $\d\eta_3 = \d x\wedge\d y$. Since $\eta_3\wedge\d\eta_3 = \frac{1}{2}\d x\wedge\d y\wedge\d z\neq 0$, we have that $\eta_3$ is a contact form in $\R^3$.

A short calculation shows that $X_1,X_2,X_3$ are contact Hamiltonian vector fields with respect to the contact structure given by $\eta_3$ with Hamiltonian functions
$$ h_1 = y\,,\quad h_2 = -x\,,\quad h_3 = -1 $$
respectively. Hence, $\langle X_1,X_2,X_3\rangle$ are also Hamiltonian vector fields relative to $(\R^3,\eta_3,X)$. Thus, the triple $(\R^3,\eta_3,X)$ is a contact Lie system with a Vessiot--Guldberg Lie algebra $\langle X_1,X_2,X_3\rangle\simeq \mathfrak{h}_3$. Moreover, the Reeb vector field is given by $Y_3$. 

The projection of the original Hamiltonian contact system (\ref{eq:Brocket-system}) onto $\mathbb{R}^2$ reads
\begin{equation}\label{Eq:ProjCon}
\frac{\d x}{\d t}=b_1(t)\,,\qquad \frac{\d y}{\d t}=b_2(t)\,,
\end{equation}
which, as foreseen by Proposition \ref{prop:conservative-contact-project-symplectic}, is Hamiltonian relative to the symplectic form $\Omega=\d x\wedge \d y$ that is determined by the condition $\d\eta=\pi^*\Omega$ for $\pi:(x,y,z)\in\mathbb{R}^3\mapsto(x,y)\in \mathbb{R}^2$. It is worth noting that the Liouville theorem for $\Omega$ on $\mathbb{R}^2$ tells us that the evolution of \eqref{Eq:ProjCon} on $\mathbb{R}^2$ leaves invariant the area of any surface, but since $\{x,y\}$ are Darboux coordinates for $\Omega$, the non-squeezing theorem also says that given a ball in $\mathbb{R}^2$ centred at a point of radius $r$, then if the image of such a ball under the dynamics of \eqref{Eq:ProjCon} is inside a ball in $\mathbb{R}^2$ of radius $R$ with center matching the center of the original ball, then $R\geq r$. In fact, the evolution of \eqref{Eq:ProjCon} is given by
$$
    x' = x+\int_0^tb_1(t')\d t'\,,\qquad y' = y+\int_0^tb_2(t')\d t'\,.
$$
Then, the image of a ball with center at a point $(x,y)$ at the time $t_0=0$ evolved relative to the evolution given by \eqref{Eq:ProjCon} until $t$ is a new ball with center at $(x',y')$ and the same radius.

It is worth noting that, by the Liouville theorem for conservative contact Lie systems, one has that the volume of a space of solutions in $\mathbb{R}^3$ does not vary on time. Hence, \eqref{eq:Brocket-system} is then a Hamiltonian system relative to a multisymplectic form $\Omega_\eta$, and therefore the  methods developed in \cite{GLRRV22} can be applied to study its properties.

\subsection{The Schwarz equation}

Consider a Schwarz equation \cite{Ber2007, Ovs2009} of the form
\begin{equation}\label{eq:Schwarz-equation}
    \frac{\d^3 x}{\d t^3} = \frac{3}{2}\left(\frac{\d x}{\d t}\right)^{-1}\left(\frac{\d^2 x}{\d t^2}\right)^2 + 2 b_1(t)\frac{\d x}{\d t}\,,
\end{equation}
where $b_1(t)$ is any non-constant $t$-dependent function. Equation \eqref{eq:Schwarz-equation} is of great relevance since it appears when dealing with Ermakov systems \cite{Lea2008} and the Schwarzian derivative \cite{Car2014}.

It is well known that equation \eqref{eq:Schwarz-equation} is a {\it higher-order Lie system} \cite{Car2012}, i.e. the associated first-order system
\begin{equation}\label{eq:ScFirst}
\frac{\d x}{\d t} = v\,,\quad \frac{\d v}{\d t} = a\,,\quad \frac{\d a}{\d t} = \frac{3}{2}\frac{a^2}{v} + 2b_1(t)v\,, 
\end{equation}
is a Lie system. Indeed, the latter system is associated with the $t$-dependent vector field $X = X_3 + b_1(t) X_1$ defined on $\mathcal{O} = \{ (x,v,a)\in\R^3\mid v\neq 0 \}$, where
$$ X_1 = 2v\parder{}{a}\,,\quad X_2 = v\parder{}{v} + 2a\parder{}{a}\,,\quad X_3 = v\parder{}{x} + a\parder{}{v} + \frac{3}{2}\frac{a^2}{v}\parder{}{a}\,. $$
These vector fields satisfy the commutation relations
$$ [X_1,X_2] = X_1\,,\quad [X_1,X_3] = 2X_2\,,\quad [X_2,X_3] = X_3\,, $$
and thus span a three-dimensional Vessiot--Guldberg Lie algebra $V = \langle X_1,X_2,X_3\rangle\simeq\mathfrak{sl}_2$.

The Schwarz equation, when written as a first-order system (\ref{eq:ScFirst}), i.e. the hereafter called {\it Schwarz system}, admits a Lie algebra of Lie symmetries, denoted by $\mathrm{Sym}(V)$, spanned by the vector fields (see \cite{LS13} for details) 
$$ Y_1 = \parder{}{x}\,,\quad Y_2 = x\parder{}{x} + v\parder{}{v} + a\parder{}{a}\,,\quad Y_3 = x^2\parder{}{x} + 2vx\parder{}{v} + 2(ax + v^2)\parder{}{a}\,. $$
These Lie symmetries satisfy the commutation relations
$$ [Y_1,Y_2] = Y_1\,,\quad [Y_1,Y_3] = 2Y_2\,,\quad [Y_2,Y_3] = Y_3\,, $$
and thus $V \simeq\Sym(V)$. The basis $\{Y_1,Y_2,Y_3\}$ admits a dual basis of one-forms $\{\eta_1,\eta_2,\eta_3\}$ given by
$$ \eta_1 = \d x - \frac{x(ax + 2v^2)}{2v^3}\d v + \frac{x^2}{2v^2}\d a\,,\quad \eta_2 = \frac{ax + v^2}{v^3}\d v - \frac{x}{v^2}\d a\,,\quad \eta_3 = -\frac{a}{2v^3}\d v + \frac{1}{2v^2}\d a\,. $$
Since
$$ \eta_2\wedge\d\eta_2 = \frac{1}{v^3}\d x\wedge\d v\wedge\d a\,, $$
we have that $(\mathcal{O},\eta_2)$ is a contact manifold. The vector fields $X_1,X_2,X_3$ are contact Hamiltonian vector fields with Hamiltonian functions
$$ h_1 = \frac{2x}{v}\,,\quad h_2 = \frac{ax - v^2}{v^2}\,,\quad h_3 = \frac{a(ax-2v^2)}{2v^3}\,, $$
respectively. Hence, $V$ consists of Hamiltonian vector fields relative to $(\mathcal{O},\eta_2)$. Thus, $(\mathcal{O},\eta_2, X)$ becomes a contact Lie system and its Reeb vector field is $Y_2$. 

Coordinates $\{x,v,a\}$ are not Darboux coordinates. Consider a new coordinate system on $\mathcal{O}$ given by
$$ q = \frac{a}{v}\,,\qquad p = \frac{x}{v}\,,\qquad z = \ln v\,. $$
Using these coordinates, $\eta_2 = \d z - p\d q$, we obtain that the Reeb vector field $Y_2$ becomes $\tparder{}{z}$, and
$$ X_1 = 2\parder{}{q}\,,\qquad X_2 = q\parder{}{q} - p\parder{}{p} + \parder{}{z}\,,\qquad X_3 = \frac{q^2}{2}\parder{}{q} + (1-pq)\parder{}{p} + q\parder{}{z}\,. $$
In Darboux coordinates $\{q,p,z\}$, the Lie symmetries $Y_1,Y_2,Y_3$ read
$$ Y_1 = \frac{1}{e^z}\parder{}{p}\,,\qquad Y_2 = \parder{}{z}\,,\qquad Y_3 = e^z\left( 2\parder{}{q} - p^2\parder{}{p} + 2p\parder{}{z} \right)\,. $$
The vector fields $X_1,X_2,X_3$ have Hamiltonian functions 
$$ h_1 = 2p\,,\qquad h_2 = pq-1\,,\qquad h_3 = \frac{1}{2}q^2 p - q\,, $$
respectively, in the given Darboux coordinates. Moreover,
$$ X = X_3 + b_1(t)X_1 = \left(\frac{q^2}{2} + 2b_1(t)\right)\parder{}{q} + (1-pq)\parder{}{p} + q\parder{}{z}\,, $$
defining the system of ordinary differential equations
\begin{equation}\label{eq:schwarz-darboux}
    \begin{dcases}
        \frac{\d q}{\d t} = \frac{q^2}{2} + 2b_1(t)\,,\\
        \frac{\d p}{\d t} = 1 - pq\,,\\
        \frac{\d z}{\d t} = q\,.
    \end{dcases}
\end{equation}
The phase portrait of system \eqref{eq:schwarz-darboux} is depicted in Figure \ref{fig:Schwarz-darboux}. It is a well-known result in contact dynamics \cite{Gas2019,DeLeo2019b} that the evolution of the Hamiltonian function along a solution is given by
$$ \Lie_{X_h}h = -(\Lie_R h)h\,, $$
where $R$ denotes the Reeb vector field. Since our Reeb vector field is $Y_2 = \tparder{}{z}$ and the Hamiltonian functions $h_1,h_2,h_3$ do not depend on the coordinate $z$, we have that our system preserves the energy along the solutions. Then, it is conservative.

\begin{figure}[ht]
    \centering
    \includegraphics[width=0.3\textwidth]{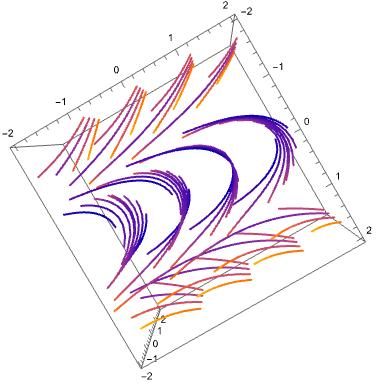}
    \includegraphics[width=0.3\textwidth]{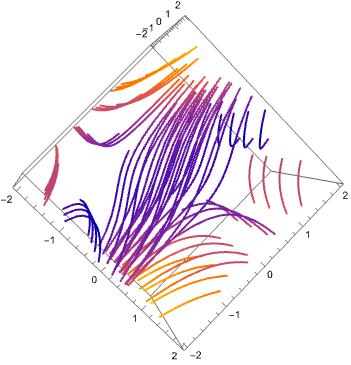}
    \includegraphics[width=0.3\textwidth]{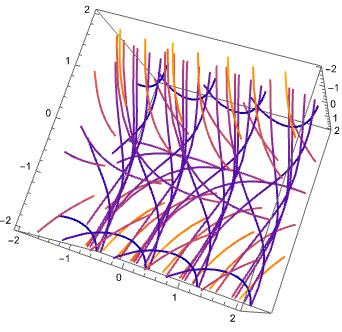}
    \caption{Phase portrait of system \eqref{eq:schwarz-darboux} from three different perspectives.}
    \label{fig:Schwarz-darboux}
\end{figure}

Note that  system \eqref{eq:schwarz-darboux} can be projected onto $\mathcal{O}/Y_2\simeq \mathbb{R}^2$, which is a consequence of Proposition \ref{prop:conservative-contact-project-symplectic}. The projected system reads
\begin{equation}\label{Eq:ProSc}
\frac{\d q}{\d t}=\frac{q^2}2+2b_1(t)\,,\qquad
\frac{\d p}{\d t}=1-pq\,,
\end{equation}
which is Hamiltonian relative to the symplectic form $\Omega=\d q\wedge \d p$. Indeed, its Hamiltonian function reads
$$
k(t,q,p)=\frac 12q^2p+2b_1(t)p\,.
$$
System \eqref{Eq:ProSc} has no equilibrium points for $b_1(t)\geq0$. Meanwhile, system \eqref{Eq:ProSc} and two equilibrium points at
$$ q = \pm 2\sqrt{-b_1(t)}\,,\quad p = \frac{\pm 1}{2\sqrt{-b_1(t)}} $$
for $b_1(t) < 0$. Setting $b_1(t) = -1/4$, system \eqref{Eq:ProSc} has the form
\begin{equation}\label{Eq:ProSc-part}
\frac{\d q}{\d t}=\frac{q^2}2-\frac{1}{2}\,,\qquad
\frac{\d p}{\d t}=1-pq\,,
\end{equation}
and has equilibrium points $(1,1)$ and $(-1,-1)$. Both equilibria are saddle points. The phase portrait for the system \eqref{Eq:ProSc-part} is depicted in Figure \ref{fig:phase-portrait-schwarz-reduced}.

\begin{figure}[ht]
    \centering
    \includegraphics[width=0.3\textwidth]{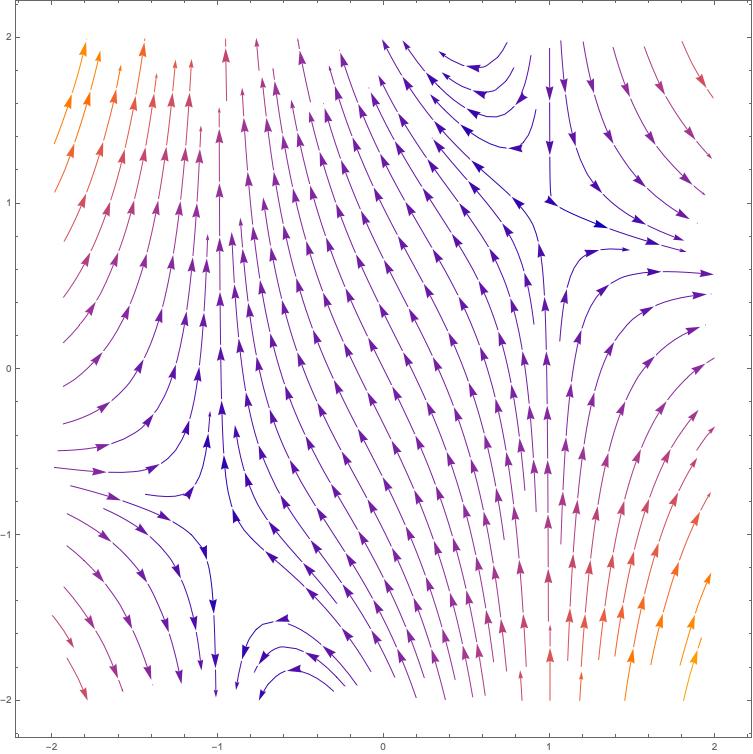}
    \caption{Phase portrait of the reduced Schwarz system \eqref{Eq:ProSc}. One can see the two saddle points at $(-1,-1)$ and $(1,1)$.}
    \label{fig:phase-portrait-schwarz-reduced}
\end{figure}

As commented in the previous section, the volume of the evolution of a ball under the dynamics of (\ref{Eq:ProSc-part}) is constant, as can be seen in Figure \ref{fig:evolution-ball-schwarz}, but if the initial ball has radius $r$, then the evolution of the ball cannot be bounded by a ball of radius smaller than $r$ with centre at the origin. 

\begin{figure}[ht]
    \centering
    \includegraphics[width=0.3\textwidth]{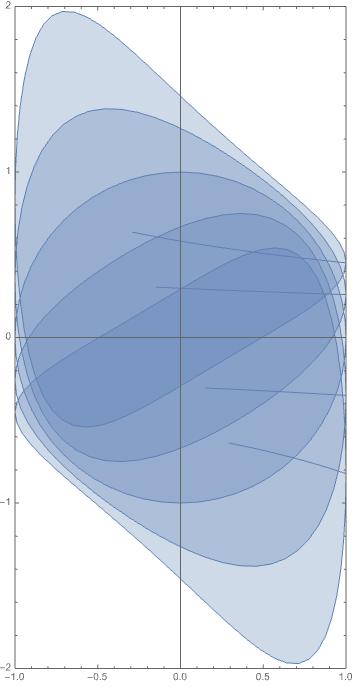}
    \caption{Evolution of a ball under the reduced Schwarz system \eqref{Eq:ProSc}. One can see that although the ball is deformed, its area is preserved.}
    \label{fig:evolution-ball-schwarz}
\end{figure}

\subsection{A quantum contact Lie system}
Let us illustrate how contact reduction can be used to reduce contact Lie systems. Consider the linear space over the real numbers, $\mathfrak{W}=\langle i\widehat{H}_1,\ldots,i\widehat{H}_5\rangle $, spanned by the basis of skew-Hermitian operators on $\mathbb{R}^2$ given by
$$ i\widehat{H}_1:=i\hat{x}\,,\quad i\widehat{H}_2:=i\hat{p}_x = \parder{}{x}\,,\quad i\widehat{H}_3:=i\hat{y}\,,\quad i\widehat{H}_4:=i\hat{p}_y = \parder{}{y}\,,\quad i\widehat{H}_5:=i\Id\,, $$
where the only non-vanishing commutation relations between the elements of the basis read 
$$ [i\widehat{H}_1,i\widehat{H}_2] = -i\widehat{H}_5\,,\quad [i\widehat{H}_3,i\widehat{H}_4] = -i\widehat{H}_5\,. $$
The Lie algebra $\mathfrak{W}$ appears in quantum mechanical problems. Let us consider the Lie algebra morphism $\rho:\mathfrak{W}\mapsto \mathfrak{X}(\mathbb{R}^5)$ satisfying that
\begin{gather*}
    \rho(i\widehat{H}_1) =: X_1 = \parder{}{x_1}\,,\qquad\rho(i\widehat{H}_2) =: X_2 = \parder{}{x_2} - x_1\parder{}{x_5}\,,\qquad \rho(i\widehat{H}_3) =: X_3 = \parder{}{x_3}\,,\\
    \rho(i\widehat{H}_4) =: X_4 = \parder{}{x_4} - x_3\parder{}{x_5}\,,\qquad \rho(i\widehat{H}_5) =: X_5 = \parder{}{x_5}\,.
\end{gather*}

Consider the Lie system on $\mathbb{R}^5$ associated with the $t$-dependent vector field
$$
    X^Q(t,x)=\sum_{\alpha=1}^5b_\alpha(t)X_\alpha(x)\,,\qquad \forall t\in \mathbb{R}\,,\quad x\in \mathbb{R}^5\,,
$$
with arbitrary $t$-dependent functions $b_1(t),\ldots,b_5(t)$, which has a Vessiot--Guldberg Lie algebra $V^Q=\langle X_1,\ldots,X_5\rangle$. The Lie algebra of Lie symmetries of $V^Q$ is spanned by the vector fields
$$\begin{gathered}
    Y_1 = \parder{}{x_1} - x_2\parder{}{x_5}\,,\qquad Y_2 = \parder{}{x_2}\,,\qquad Y_3 = \parder{}{x_3} - x_4\parder{}{x_5}\,,\\
    Y_4 = \parder{}{x_4}\,,\qquad Y_5 = \parder{}{x_5}\,.
\end{gathered}
$$
Since $Y_1,\wedge \ldots \wedge Y_5\neq 0$ at every point of $\mathbb{R}^5$, there exists a basis of differential one-forms  on $\mathbb{R}^5$ dual to $Y_1,\ldots,Y_5$ given by
$$ \eta_1 = \d x_1\,,\quad\eta_2 = \d x_2\,,\quad\eta_3 = \d x_3\,,\quad\eta_4 = \d x_4\,,\quad\eta_5 = \d x_5 + x_2\d x_1 + x_4\d x_3\,, $$
i.e. $\eta_i(Y_j)=\delta_{ij}$, for $i,j=1,\ldots,5$, where $\delta_{ij}$ is the Kronecker's delta function. 
Then, $\eta_5\wedge (\d\eta_5)^2=2\d x_1\wedge\d x_2\wedge\d x_3\wedge\d x_4\wedge\d x_5$ is a volume form on $\R^5$ and thus $\eta_5$ becomes a contact form on $\R^5$. Moreover, $X_1,X_2,X_3,X_4,X_5$ are contact Hamiltonian vector fields with Hamiltonian functions
$$
h_1 = -x_2\,,\quad h_2 = x_1\,,\quad h_3 = -x_4\,,\quad h_4 = x_3\,,\quad h_5 = -1\,,
$$
respectively. Thus, $(\R^5, \eta_5, X^Q)$ admits a Vessiot--Guldberg Lie algebra $V^Q$ of Hamiltonian vector fields relative to $\eta_5$ and $(\R^5, \eta_5, X^Q)$ becomes a contact Lie system. The Reeb vector field of $\eta_5$ is given by $X_5=Y_5$. Since the Hamiltonian functions $h_1,\ldots,h_5$ are first-integrals of the Reeb vector field, $(\mathbb{R}^5,\eta_5,X)$ is a conservative contact Lie system. It is relevant that many important techniques for studying contact Lie system will be available only for conservative contact Lie systems. 

Let us consider the Lie algebra of symmetries of  $V^Q$ spanned by
$$
    V^S = \langle Y_1,Y_2,Y_5\rangle.
$$
This Lie algebra is isomorphic to the Heisenberg three-dimensional Lie algebra $\mathfrak{h}_3$. Moreover, the vector fields of $V^S$ are also Hamiltonian relative the contact structure $\eta_5$.
The momentum map $J:\mathbb{R}^5\rightarrow \mathfrak{h}_3^*$ associated with $V^S$ is such that $i_{X_i}\eta_5=J^i$ for $i=1,2,5$, where 
$$
J^1=x_2,\qquad J^2=-x_1,\qquad J^5=-1,\qquad 
$$
Note that $J$ is not a submersion, but its tangent map has constant rank. By the Constant Rank Theorem, $J^{-1}(\mu)$ is a submanifold for every $\mu\in \mathfrak{h}_3^*$ and the tangent space at one of its points is given by the kernel of $\T_pJ$, whatever $\mu\in \mathfrak{h}_3^*$ is. By Theorem \ref{thm:reduction-willet}, the submanifold $J^{-1}(\mathbb{R}_+\mu)$ is invariant relative to the evolution of the contact Lie system. 

Let us give the integral curves of the vector fields $X_1,X_2,X_5$:
\begin{align*}
    & X_1\quad\longrightarrow\quad x_1'=x_1+\lambda_1\,,&& x_2'=x_2\,,&& x_3'=x_3\,,&& x_4'=x_4\,,&& x_5'=x_5\,,\\
    & X_2\quad\longrightarrow\quad x_1'=x_1\,,&& x_2'=x_2+\lambda_2\,,&& x_3'=x_3\,,&& x_4'=x_4\,,&& x_5'=x_5-\lambda_2x_1\,,\\
    & X_5\quad\longrightarrow\quad x_1'=x_1\,,&& x_2'=x_2\,,&& x_3'=x_3\,,&& x_4'=x_4\,,&& x_5'=x_5+\lambda_3\,,
\end{align*}
where $\lambda_1,\lambda_2,\lambda_3\in\R$. Therefore, $\mathscr{L}_{X_5}{ J}=0$ and
$
\lambda(\mu_1,\mu_2,-1)=(\lambda\mu_1,\lambda \mu_2,-\lambda)\notin {\rm Im}\,J$ unless $\lambda=1$. Then, $J^{-1}(\mathbb{R}_+\mu)=J^{-1}(\mu)=\{x_1,x_2\}\times \mathbb{R}^3$.
Moreover,
$$
J^{-1}(\mu)/G_5=\{x_1,x_2\}\times \mathbb{R}^3\,. 
$$
Therefore, $J^{-1}(\mu)/G_5$ admits coordinates $\{x_3,x_4,x_5\}$. Note that the projection of the initial contact Lie system onto $J^{-1}(\mu)/G_5$ reads
$$
    \bar X^Q(t,x)=\sum_{\alpha=2}^5b_\alpha(t)\widehat X_\alpha(x),\qquad \forall t\in \mathbb{R}\,,\qquad \forall x\in \mathbb{R}^3,
$$
while the projection of the initial Vessiot--Guldberg Lie algebra is spanned by the vector fields
\begin{gather*}
    \widehat X_2=-x_1\frac{\partial}{\partial x_5},\qquad \widehat X_3 = \parder{}{x_3}\,,\qquad
    \widehat X_4 = \parder{}{x_4} - x_3\parder{}{x_5}\,,\qquad  \widehat X_5 = \parder{}{x_5}\,.
\end{gather*}
These are Hamiltonian vector fields relative to the contact form $\d x_5 + x_4\d x_3$ with Hamiltonian functions
$$
\bar h_2= x_1,\qquad \bar h_3 = -x_4\,,\qquad \bar h_4 = x_3\,,\qquad \bar h_5 = -1\,.
$$
Since $X_5$ is the Reeb vector fields on $\mathbb{R}^3$ relative to $\d x_5+x_4\d x_3$, the reduced contact Lie system is also conservative. In fact, it could be projected onto $\mathbb{R}^3/X_5\simeq \mathbb{R}^2$, giving rise to a Lie--Hamilton system on $\mathbb{R}^2$ of the form
$$
\frac{\d x_3}{\d t}=b_3(t),\qquad \frac{\d x_4}{\d t}=b_4(t)
$$
relative to $\omega=\d x_4\wedge \d x_3$.

\subsection{A non-conservative example}

Consider the manifold $M = \R^3$ equipped with linear coordinates $\{q,p,z\}$. The manifold $M$ has a natural contact structure given by the one-form $\eta = \d z - p\,\d q$. Its associated Reeb vector field is $R = \tparder{}{z}$. Consider the vector fields on $M$ given by
$$ X_1 = \parder{}{z}\,,\qquad X_2 = \parder{}{q}\,,\qquad X_3 = z\parder{}{q} - p^2\parder{}{p}\,. $$
These vector fields are Hamiltonian relative to $(\R^3,\eta)$ with Hamiltonian functions
$$ h_1 = -1\,,\qquad h_2 = p\,,\qquad h_3 = pz\,, $$
and span a three-dimensional Lie algebra with commutation relations
$$ [X_1,X_2] = 0\,,\qquad [X_1,X_3] = X_2\,,\qquad [X_2,X_3] = 0\,, $$
isomorphic to $\mathfrak{h}_3$. 
This allows us to define a contact Lie system on $\mathbb{R}^3$ relative to $\eta$ given by
\begin{equation}\label{eq:vector-field-example-nonconservative}
    X=\sum_{\alpha=1}^3b_\alpha(t)X_\alpha\,.
\end{equation}
where $b_1(t),b_2(t),b_3(t)$ are arbitrary $t$-dependent functions.
Since the Hamiltonian function of $X_3$ is not a first-integral of the Reeb vector field $R$, then $X$ is a non-conservative contact Lie system. Note also that $X$ is associated with the $t$-dependent Hamiltonian function
$$
    h = \sum_{\alpha = 1}^3 b_\alpha(t)h_\alpha\,,
$$
namely each $X_t$ is the Hamiltonian vector field related to $h_t$ for every $t\in \mathbb{R}$. As a consequence, the volume form related to the contact form, namely
$$
\Omega = \d\eta\wedge \eta = \d q\wedge \d p\wedge \d z\,,
$$
is not invariant relative to the vector fields of the Vessiot--Guldberg Lie algebra and $\Omega$ is not, in general, preserved by the evolution of $X$. More specifically, if $F_{t_0}:\R\times \R^3\rightarrow \R^3$ is the flow starting from the point $t_0$ of $X$, namely $F_{t_0}(t_0,x_0)=x_0$   and $F_{t_0}(t,x_0) = x(t)$, where $x(t)$ is the particular solution to \eqref{eq:vector-field-example-nonconservative} with $x(t_0) = x_0$, for every $x_0\in \R^3$, then
$$
\frac{\d}{\d t}\int_{F_{t_0}(t,A)}\Omega=\frac{\d}{\d t}\int_AF^*_{t_0,t}\Omega=\int_A\Lie_X\Omega\,,
$$
for every subset $A\subset \R^3$.
But $\Lie_X\Omega = 2(Rh) \Omega$. Hence,
$$
\frac{\d}{\d t}\int_{F_{t_0}(t,A)}\Omega = 2\int_A(Rh)\Omega = 2\int_A\left(\sum_{\alpha=1}^3b_\alpha(t)Rh_\alpha \right)\Omega\,.
$$
Note that if $V=V^X$, then
$$
\sum_{\alpha=1}^3b_\alpha(t)Rh_\alpha = b_3(t)p\neq0\,
$$
for a generic value of $p\in \mathbb{R}$ and $t\in\mathbb{R}$. 

\section{Coalgebra method to obtain superposition rules of Jacobi--Lie systems}\label{Se:Coalgebra}

Let us provide a method to derive superposition rules for contact Lie systems via Poisson coalgebras. Our method is a modification of the coalgebra method for deriving superposition rules for Dirac--Lie systems devised in \cite{Car2014}. It is worth noting that the coalgebra method does not work for contact Lie systems {\it per se} since, as proved next, the diagonal prolongations of a contact Lie system will not always be a contact Lie system.


Let us start by defining the diagonal prolongation of the sections of a vector bundle, as this is a key for developing the coalgebra method for Lie systems admitting Vessiot--Guldberg Lie algebras of Hamiltonian vector fields relative to different geometric structures.

The {\it diagonal prolongation to $M^k$ of a vector bundle} $\tau:F\to M$ on $M$ is defined to be the Whitney sum of $k$-times the vector bundle $\tau$ with itself, namely the vector bundle $\tau^{[k]}:F^k = F\times\overset{(k)}{\dotsb}\times F\mapsto M^k= M\times \overset{(k)}{\dotsb} \times M$, viewed as a vector bundle over $M^k$ in the natural way, i.e.
\begin{equation*}
   \tau^{[k]}(f_{(1)},\dotsc,f_{(k)}) = (\tau(f_{(1)}),\dotsc,\tau(f_{(k)}))\,,\qquad \forall f_{(1)},\ldots,f_{(k)}\in F\,, 
\end{equation*}
and
\begin{equation*}
    F^k_{(x_{(1)},\ldots,x_{(k)})}=F_{x_{(1)}}\oplus\cdots\oplus F_{x_{(k)}}\,,\qquad \forall (x_{(1)},\ldots,x_{(k)})\in M^k\,.
\end{equation*}

Every section $e:M\to F$ of the vector bundle $\tau$ has a natural {\it  diagonal prolongation} to a section $e^{[k]}$ of the vector bundle $\tau^{[k]}$ given by
\begin{equation*}
    e^{[k]}(x_{(1)},\ldots,x_{(k)})=e(x_{(1)})+\cdots +e(x_{(k)})\,,\qquad \forall (x_{(1)},\ldots,x_{(k)})\in M^k.
\end{equation*}
The {\it  diagonal prolongation of a function $f\in \Cinfty (M)$} to $M^k$ is the function
$f^{[k]}(x_{(1)},\ldots,x_{(k)})= f(x_{(1)})+\ldots+f(x_{(k)})$. Consider also the sections $e^{(j)}$ of $\tau^{[k]}$, where $j\in \{1,\ldots,n\}$ and $e$ is a section of $\tau$, given by
\begin{equation}\label{prol1}
e^{(j)}(x_{(1)},\dots,x_{(k)})=0+\cdots +e(x_{(j)})+\cdots+0\,,\qquad \forall (x_{(1)},\ldots,x_{(k)})\in M^k\,.
\end{equation}
If $\{e_1,\ldots, e_r\}$ is a basis of local sections of the vector bundle $\tau$, then $e_i^{(j)}$, with $j = 1,\dotsc,k$ and $i = 1,\dotsc,r$, is a basis of local sections of $\tau^{[k]}$.

Due to the obvious canonical isomorphisms
$$(\T M)^{[k]}\simeq \T M^k\quad\text{and}\quad (\cT M)^{[k]}\simeq \cT M^k\,,$$
the diagonal prolongation $X^{[k]}$ of a vector field $X\in\X(M)$ can be understood as a vector field $\widetilde{X}^{[k]}$ on $M^k$, and the diagonal prolongation, $\alpha^{[k]}$, of a one-form $\alpha$ on $M$ can be understood as a one-form $\widetilde{\alpha}^{[k]}$ on $M^k$. If $k$ is assumed to be fixed, we will simply write $\widetilde{X}$ and $\widetilde{\alpha}$ for their diagonal prolongations. 

The proofs of Proposition \ref{Pro:ProlJL}, its Corollaries \ref{Cor:LieAlgJL} and \ref{Cor:LieAlgJL2}, and Proposition \ref{Prop:Obv2} below are straightforward as they rely, almost entirely, on the definition of diagonal prolongations. Anyhow, as Jacobi manifolds with a non-vanishing Reeb vector field give rise to a Dirac manifold, they can also be considered as particular cases of the results given for Dirac structures in \cite{Car2014}. 

\begin{proposition}\label{Pro:ProlJL} Let $(M,\Lambda,E)$ be a Jacobi manifold with bracket $\{\cdot,\cdot\}_{\Lambda,E}$. Let $X$ and $f$ be a vector field and a function on $N$. Then:
\begin{enumerate}[(a)]
    \item $(M^k,\Lambda^{[k]},E^{[k]})$ is a Jacobi manifold for every $k\in \mathbb{N}$.
    \item If $f$ is a Hamiltonian function for a Hamiltonian vector field $X$ relative to $(M,\Lambda,E)$, its diagonal prolongation $f^{[k]}$ to $M^k$ is a Hamiltonian function of the diagonal prolongation, $X^{[k]}$, to $M^k$ with respect to $(M^k,\Lambda^{[k]},E^{[k]})$. 
    \item If $f \in {\rm Cas}(M ,\Lambda, E)$, then $f^{[k]}\in {\rm Cas}(M^k,\Lambda^{[k]}, E^{[k]})$.
    \item The map $\lambda: (\Cinfty(M),\{\cdot,\cdot\}_{\Lambda,E})\ni f \mapsto f^{[k]}\in (\Cinfty(M^k), \{\cdot,\cdot \}_{\Lambda^{[k]},E^{[k]}})$ is an injective Lie algebra morphism.
\end{enumerate}
\end{proposition}
 
\begin{corollary}\label{Cor:LieAlgJL}
    If $h_1,\dotsc, h_r:M\to \mathbb{R}$ is a family of functions on a Jacobi manifold $(M,\Lambda,E)$ spanning a finite-dimensional real Lie algebra of functions with respect to the Lie bracket $\{\cdot,\cdot\}_{\Lambda,E}$, then their diagonal prolongations $\tilde h_1, \dotsc,\tilde h_r$ to $M^k$ close an isomorphic Lie algebra of functions with respect to the Lie bracket $\{\cdot,\cdot\}_{\Lambda,E}$ induced by the Jacobi manifold $(M^k,\Lambda^{[k]}, E^{[k]})$.
\end{corollary}

\begin{corollary}\label{Cor:LieAlgJL2}
    Let $(M, \Lambda, E, X)$ be a Jacobi--Lie system admitting a Lie--Hamiltonian $(M, \Lambda, E, h)$. Then, the tuple $(M^k,\Lambda^{[k]}, E^{[k]},X^{[k]})$ is a Jacobi--Lie system admitting a Jacobi--Lie Hamiltonian of the form $(M^k,\Lambda^{[k]}, E^{[k]},h^{[k]})$, where $h_t^{[k]} = \tilde h_t^{[k]}$ is the diagonal prolongation of $h_t$ to $M^k$.
\end{corollary}

\begin{proposition}\label{Prop:Obv2}
    If $X$ be a system possessing a $t$-independent constant of the motion $f$ and $Y$ is a $t$-independent Lie symmetry of $X$, then:
    \begin{enumerate}
        \item The diagonal prolongation $f^{[k]}$ is a $t$-independent constant of the motion for $X^{[k]}$.
        \item Then $Y^{[k]}$ is a $t$-independent Lie symmetry of $X^{[k]}$.
        \item If $h$ is a $t$-independent constant of the motion for $X^{[k]}$, then $Y^{[k]}h$ is another $t$-independent constant of the motion for $X^{[k]}$.
    \end{enumerate}
\end{proposition}
\begin{proposition} The diagonal prolongation to $M^k$ of a Jacobi--Lie system $(M,\Lambda,E,X)$ is a Jacobi--Lie system $(M,\Lambda^{[k]},E^{[n]},X^{[k]})$.
\end{proposition}

For the sake of completeness, let us prove the following result.

\begin{proposition}
    Let $(M,\Lambda,E,X)$ be a Jacobi--Lie system possessing a Jacobi--Lie Hamiltonian $(M, \Lambda,E, h)$. A function $f\in \Cinfty(M)$ is a constant of the motion for $X$ if and only if it commutes with all the elements of ${\rm Lie}(\{h_t\}_{t\in \R},\{\cdot,\cdot\})$.
\end{proposition}
\begin{proof} The function $f$ is a constant of the motion for $X$ if
\begin{equation}\label{Eq:Cons}
0=X_tf=\{h_t,f\}\,,\qquad \forall t\in \mathbb{R}\,.
\end{equation}
Hence, 
$$
\{f,\{h_t,h_{t'}\}\}=\{f,\{h_t,h_{t'}\}\}+\{f,\{h_t,h_{t'}\}\}\,,\qquad \forall t,t'\in \mathbb{R}\,.
$$
Inductively, $f$ is shown to commute with all the elements of ${\rm Lie}(\{h_t\}_{t\in \mathbb{R}})$. 
Conversely, if $f$ commutes with all ${\rm Lie}(\{h_t\}_{t\in \mathbb{R}})$ relative to $\{\cdot,\cdot\}$, in particular, \eqref{Eq:Cons} holds and $f$ is a constant of the motion of $X$.
\end{proof}

The bracket for Jacobi--Lie systems is not a Poisson bracket. It becomes only a Poisson bracket for good Hamiltonian functions. Nevertheless, when a Lie group action gives rise to a momentum map, the components of the momentum map are first-integrals of $R$. As a consequence, the following proposition, which can be considered as an adaptation of \cite[Proposition 8.4]{Car2014}, is satisfied. Recall that if $(M,\Lambda,E)$ is a Jacobi manifold, then $\Cinfty(M^k)$ becomes a Lie algebra relative to the Lie bracket $\{\cdot,\cdot\}_k$ related to $\Lambda^{[k]}$ and $\Cinfty(\mathfrak{W}^*)$ is a Poisson algebra relative to the Kirillov--Kostant--Souriau bracket.

\begin{proposition}
    Given a Jacobi--Lie system $(M,\Lambda,R,X)$ that admits a Jacobi--Lie Hamiltonian $(M,\Lambda, R,h)$ such that $\{h_t\}_{t\in\R}$ is contained in a finite-dimensional Lie algebra of functions $(\mathfrak{W}, \{\cdot,\cdot\})$. Given the good momentum map $J : M\to \mathfrak{W}^\ast$ associated with a contact Lie group action leaving $h$ invariant, the pull-back $J^\ast(C)$ of any Casimir function $C$ on $\mathfrak{W}^\ast$ is a constant of the motion for $X$. Moreover, if $C=C(v_1,\ldots,v_r)$, where $v_1,\ldots v_r$ is a basis of linear coordinates on $\mathfrak{W}^*$, then 
    \begin{equation}\label{Eq:ExCas}
C\left(\sum_{a=1}^kh_1(x_{(a)}),\ldots, \sum_{a=1}^kh_r(x_{(a)})\right)\,, \qquad J^*v_i=h_i\,,\qquad i=1,\ldots,r\,,
\end{equation} is a constant of the motion of $X^{[k]}$.
\end{proposition}

The coalgebra method takes its name from the fact that it analyses the use of Poisson coalgebras and a so-called {\it coproduct} to obtain superposition rules. In fact, the coproduct is responsible for the form of \eqref{Eq:ExCas}.

Finally, let us provide an example of the coalgebra method for contact Lie systems. Let us consider the Lie group $\SL(2,\mathbb{R})$ of $2\times 2$ matrices with determinant one and real entries, i.e.
$$
    \SL(2,\R) = \left\{\begin{pmatrix}\alpha&\beta\\\gamma&\delta\end{pmatrix}\ \bigg\vert\  \alpha\delta - \beta\gamma = 1\right\}\,,
$$
and the automorphic Lie system
\begin{equation}\label{Eq:AusSL2}
    \frac{\d g}{\d t}=\sum_{\alpha = 1}^3b_\alpha(t)X_\alpha^R(g)\,,\qquad \forall t\in \mathbb{R}\,,\quad\forall g\in \SL(2,\R)\,,
\end{equation}
where $b_1(t),b_2(t),b_3(t)$ are arbitrary $t$-dependent functions. Observe that $\alpha,\beta,\gamma$ become a coordinate system of $\SL(2,\R)$ close to the identity. The Lie algebra of right-invariant vector fields on $\SL(2,\R)$ is spanned by the vector fields
$$
    X^R_1=\alpha\frac{\partial}{\partial\alpha}+\beta\frac{\partial}{\partial \beta}-\gamma\frac{\partial}{\partial \gamma}\,,\qquad X^R_2=\gamma\frac{\partial}{\partial \alpha}+\frac{1+\beta\gamma}{\alpha}\frac{\partial}{\partial \beta}\,,\qquad X^R_3=\alpha\frac{\partial}{\partial \gamma}\,,
$$
and their commutation relations are
$$ [X_1^R,X_2^R] = -2X_2^R\,,\qquad [X_2^R,X_3^R] = -X_1^R\,,\qquad [X_1^R,X_3^R] = 2X_3^R\,. $$
Meanwhile, the left-invariant vector fields are spanned by
$$
    X^L_1=\alpha\frac{\partial}{\partial\alpha}-\beta\frac{\partial}{\partial \beta}+\gamma\frac{\partial}{\partial \gamma}\,,\qquad X^L_2=\alpha\frac{\partial}{\partial \beta}\,,\qquad X^L_3=\beta\frac{\partial}{\partial \alpha}+\frac{1+\beta\gamma}{\alpha}\frac{\partial}{\partial \gamma}\,.
$$
Moreover,
$$
    [X_1^L,X_2^L] = 2X_2^L\,,\qquad [X_2^L,X_3^L]=X_1^L\,,\qquad [X_1^L,X_3^L]=-2X_3^L\,.
$$
Consider the set of 
the left-invariant differential forms on $\SL(2,\mathbb{R})$ given by 
$$
    \eta^L_1 = \frac{1+\beta\gamma}{\alpha}\d\alpha-\beta \d\gamma\,,\qquad \eta^L_2=\frac{\beta(1 + \beta\gamma)}{\alpha^2} \d\alpha+\frac1\alpha \d\beta-\frac{\beta^2}{\alpha}\d\gamma\,,\qquad \eta_3^L=-\gamma \d\alpha+\alpha \d\gamma\,,
$$
which become a basis of the space of left-invariant differential forms on $\SL(2,\mathbb{R})$. It is relevant that
$$
 \d\eta^L_1=\eta^L_2\wedge \eta^L_3\Rightarrow   \d\eta_1^L\wedge \eta_1^L\neq 0\,.
$$
Hence, $\eta_1$ becomes a left-invariant contact form on $\SL(2,\mathbb{R})$ with a Reeb vector field $X_1^L$. Therefore, the vector fields $X_1^R,X_2^R,X_3^R$ admit the Hamiltonian functions
$$
    h_1=-\eta_1^L(X_1^R) = -1 - 2\beta\gamma\,,\qquad h_2 = -\eta_1^L(X_2^R)=-\frac{\gamma}{\alpha}(1+\beta\gamma)\,,\qquad h_3=-\eta_1^L(X_3^R)=\alpha\beta\,.
$$
These Hamiltonian functions satisfy the commutation relations
$$
    \{h_1,h_2\}=-2h_2\,,\qquad \{h_1,h_3\}=2h_3\,,\qquad \{h_2,h_3\}=-h_1\,.
$$
Hence, all Hamiltonian functions for the right-invariant vector fields relative to the contact form $\eta^L_1$ are first-integrals of the Reeb vector field of $\eta^L_1$, namely $X_1^L$. This can be used to obtain the superposition rule for Lie systems on $\SL(2,\mathbb{R})$. Let us explain this. Let $\{e^1,e^2,e^3\}$ be a basis of $\mathfrak{sl}_2^*$ dual to $\{X_1^L(e),X_2^L(e),X_3^L(e)\}$. Given the action of $G$ on itself on the left, whose fundamental vector fields are given by the linear space of right-invariant vector fields on $\SL(2,\mathbb{R})$, one may define an associated momentum map
$$
    J:A\in \SL(2,\mathbb{R})\longmapsto -(1+2\beta\gamma) e^1-\frac{\gamma}{\alpha}(1+\beta\gamma)e^2+\alpha\beta e^3\in\mathfrak{sl}_2^*\,.
$$
This allows us to obtain a superposition rule using the coalgebra method. The theory of Lie systems states that, in order to determine a superposition rule for a Lie system, one has to determine the smallest $k\in\N$ so that the vector fields $[X_1^R]^{[k]},[X_2^R]^{[k]},[X_3^R]^{[k]}$ will be linearly independent at a generic point (see \cite{LS2020}). Since $X_1^R,X_2^R,X_3^R$ are linearly independent at every point of $\SL(2,\mathbb{R})$, it follows that $m=1$. Hence, a superposition rule for \eqref{Eq:AusSL2} can be obtained by deriving three common first-integrals for $[X_1^R]^{[m+1]},[X_2^R]^{[k+1]},[X_3^R]^{[k+1]}$, let us say $I_1,I_2,I_3$, satisfying
$$
\frac{\partial(I_1,I_2,I_3)}{\partial(\alpha,\beta,\gamma)}\neq0\,.
$$
A good Hamiltonian function that Poisson commutes with $h_1,h_2,h_3$ is given by
$$
    C_1 = 4h_2(\alpha,\beta,\gamma)h_3(\alpha,\beta,\gamma) + h_1(\alpha,\beta,\gamma)^2\in\Cinfty(\SL(2,\R))\,,
$$
where $\alpha,\beta,\gamma$ are assumed to be functions on $\SL(2,\R)$. Similarly,
$$
\begin{gathered}
    h_1^{[2]} = -(1 + 2\beta\gamma) - (1 + 2\beta'\gamma')\,,\qquad h_2^{[2]} = -\frac{\gamma}{\alpha}(1+\beta\gamma)-\frac{\gamma'}{\alpha'}(1+\beta'\gamma')\,,\\
    h_3^{[2]} = \alpha\beta + \alpha'\beta'
\end{gathered}
$$
become the Hamiltonian functions of
$$
\begin{gathered}
    [X_1^R]^{[2]} = \alpha\frac{\partial}{\partial\alpha}-\beta\frac{\partial}{\partial \beta}+\gamma\frac{\partial}{\partial \gamma} + \alpha'\frac{\partial}{\partial\alpha'}-\beta'\frac{\partial}{\partial \beta'}+\gamma'\frac{\partial}{\partial \gamma'}\,,\qquad
    [X^R_2]^{[2]} = \alpha\frac{\partial}{\partial \beta} + \alpha'\frac{\partial}{\partial \beta'}\,,\\
    [X_3^R]^{[2]} = \beta\frac{\partial}{\partial \alpha}+\frac{1+\beta\gamma}{\alpha}\frac{\partial}{\partial \gamma} + \beta'\frac{\partial}{\partial \alpha'}+\frac{1+\beta'\gamma'}{\alpha'}\frac{\partial}{\partial \gamma'}\,.
\end{gathered}
$$
Hence, a common first-integral for $[X_1^R]^{[2]},[X_2^R]^{[2]},[X_3^R]^{[2]}$ is given by
$$
    I_1 = 4h_2^{[2]}h_3^{[2]} - (h_1^{[2]})^2 = -\frac{4 (\beta \gamma \alpha' + \alpha'-\alpha \beta \gamma') (\gamma \alpha' \beta'-\alpha (\beta' \gamma'+1))}{\alpha \alpha'}\,.
$$
Note that this is indeed an application of \eqref{Eq:ExCas} to our problem.

To obtain the remaining two first-integrals for $[X_1^R]^{[2]},[X_2^R]^{[2]},[X_3^R]^{[2]}$, we derive
$$
    I_2 = [X^L_2]^{[2]} I_1 = -\frac{4(\gamma\alpha' - \alpha\gamma') \big((1 + \beta \gamma) \alpha'^2 - 
  \alpha (\alpha - \gamma \alpha' \beta' + \beta \alpha' \gamma' + \alpha \beta' \gamma')\big) }{\alpha\alpha'}\,,
$$
\begin{align*}
    I_3 &= [X^L_3]^{[2]} I_1 \\
    &= -\frac{4 (\alpha \beta (\beta' \gamma'+1)-(\beta \gamma+1) \alpha' \beta') \left(\alpha (\beta' \gamma' \alpha+\alpha-\gamma \alpha' \beta'+\beta \alpha' \gamma')-(\beta \gamma+1) \alpha'^2\right)}{\alpha^2 \alpha'^2}\,.
\end{align*}
Since the determinant of
$$
    \frac{\partial (I_1,I_2,I_3)}{\partial (\alpha,\beta,\gamma)} = \begin{pmatrix}
        \dparder{I_1}{\alpha} & \dparder{I_1}{\beta} & \dparder{I_1}{\gamma} \\\\
        \dparder{I_2}{\alpha} & \dparder{I_2}{\beta} & \dparder{I_2}{\gamma} \\\\
        \dparder{I_3}{\alpha} & \dparder{I_3}{\beta} & \dparder{I_3}{\gamma}
    \end{pmatrix}
$$
is different from zero at a generic point in $\SL(2,\mathbb{R})\times \SL(2,\mathbb{R})$, the system of algebraic equations
\begin{equation}\label{Eq:SupAlg}
        I_1 = \lambda_1\,,\qquad
        I_2 = \lambda_2\,,\qquad
        I_3 = \lambda_3\,,
\end{equation}
allows us to obtain $\alpha,\beta,\gamma$ in terms of $\alpha',\beta',\gamma'$ and $\lambda_1,\lambda_2,\lambda_3$, which gives rise to a superposition rule. Its expression may be complicated, but can be derived using any program of mathematical manipulation. 

Anyway, there is a simpler method to obtain the superposition rule for \eqref{Eq:AusSL2}. Since it is an automorphic Lie system with a Vessiot--Guldberg Lie algebra of right-invariant vector fields, it is known that a superposition rule is given by the multiplication on the right
$$
\Phi:(g,h)\in G\times G\mapsto gh\in G\,.
$$
Since the vector fields span a distribution of dimension three on $\SL(2,\R)\times \SL(2,\R)$, which is three-codimensional, it was proved in \cite{CGM07} that the superposition rule must be unique. Hence, this superposition rule must be the one obtained by solving the algebraic system \eqref{Eq:SupAlg}.
\section{Conclusions and further research}

In this paper, we have introduced the notion of contact Lie system: systems of first-order differential equations describing the integral curves of a $t$-dependent vector field taking values in a finite-dimensional Lie algebra of Hamiltonian vector fields relative to a contact manifold. In particular, we have studied families of conservative contact Lie systems, i.e. being invariant relative to the flow of the Reeb vector field. We have also developed Liouville theorems, a contact reduction and a Gromov non-squeezing theorems for certain classes of contact Lie systems. We have also classified locally transitive contact Lie systems on three-dimensional manifolds. In order to illustrate these results, we have worked out several examples, such as the Brockett control system, the Schwarz equation, an automorphic Lie system on $\SL(2,\R)$, and a quantum contact Lie system.

The reduction procedures developed by Willet \cite{Wil2002} and Albert \cite{Alb1989} and the one introduced in this paper open the door to develop an energy-momentum method \cite{Mars1988} for contact Lie systems, both conservative and non-conservative. This will allow us to study the relative equilibria points of these systems. We also believe that a new type of contact reduction can be achieved by interpreting contact forms in a new manner. This is currently being developed and, hopefully, will be published in a future work.  

Recently, the contact formulation for non-conservative mechanical systems has been generalised via the so-called $k$-contact \cite{Gas2020,Gas2021,Gra2022}, $k$-cocontact \cite{Ri-2022}, and multicontact \cite{LGMRR-2022,Vi-2015} formulations. It would be interesting to study the Lie systems whose Vessiot--Guldberg Lie algebra consists of Hamiltonian vector fields relative to these structures. It would also be interesting to classify contact Lie systems possessing a transitive primitive Vessiot--Guldberg Lie algebra \cite{Shn-1984,Shn-1984b}.

\section*{Acknowledgments}

X. Rivas acknowledges financial support from the Ministerio de Ciencia, Innovaci\'on y Universidades (Spain), projects PGC2018-098265-B-C33 and D2021-125515NB-21.
J. de Lucas and X. Rivas acknowledge partial financial support from the Novee Idee 2B-POB II project PSP: 501-D111-20-2004310 funded by the ``Inicjatywa Doskonałości - Uczelnia Badawcza'' (IDUB) program.

\bibliographystyle{WAW-BCN}
\bibliography{bibliografia_old.bib}


\end{document}